\documentclass[USenglish,oneside,twocolumn]{article}

\usepackage[utf8]{inputenc}
\usepackage[big]{dgruyter_NEW}
\usepackage{hyperref}
\usepackage{xcolor}
\usepackage{amssymb,amsmath,amsthm}
\usepackage{bm}
\usepackage{caption}
\usepackage{subcaption}
\usepackage[ruled]{algorithm2e}

\SetCommentSty{mycommfont}
\usepackage{algorithmic}
\usepackage{balance}

\newcommand{\AsymML}{AsymML}

\newcommand{\conv}{\text{Conv}}
\newcommand{\convtrust}{\text{Conv}_\text{T}}
\newcommand{\convuntrust}{\text{Conv}_\text{U}}

\newcommand{\ceiling}[1]{\left \lceil #1 \right \rceil}
\newcommand{\norm}[1]{\left \| #1 \right \|}

\newcommand{\XT}{X^{\text{(T)}}}
\newcommand{\XUT}{X^{\text{(U)}}}
\newcommand{\XUTn}{X^{'\text{(U)}}}
\newcommand{\DUT}{\mathcal{M}(\mathcal{X}^{(U)})}

\newcommand{\YT}{Y^{\text{(T)}}}
\newcommand{\YUT}{Y^{\text{(U)}}}
\newcommand{\dX}{\nabla_{X} \mathcal{L}}
\newcommand{\dY}{\nabla_{Y} \mathcal{L}}
\newcommand{\dW}{\nabla_{W_{i,j}} \mathcal{L}}
\newcommand{\dWT}{\nabla_{W_{i,j}}^{\text{(T)}} \mathcal{L}}
\newcommand{\dWUT}{\nabla_{W_{i,j}}^{\text{(U)}} \mathcal{L}}

\newcommand{\ui}[1]{\bm{u}^{(#1)}}
\newcommand{\vi}[1]{\bm{v}^{(#1)}}
\newcommand{\Xin}[1]{\overline{X}^{(#1)}}
\newcommand{\XTfl}{\overline{X}^{\text{(T)}}}
\newcommand{\XUfl}{\overline{X}^{\text{(U)}}}

\newcommand{\ci}{N}
\newcommand{\cii}{\mathcal{R}}
\newcommand{\co}{M}
\newcommand{\Hi}{h}
\newcommand{\Wi}{w}
\newcommand{\Ho}{h^{'}}
\newcommand{\Wo}{w^{'}}

\newcommand{\jj}{p}

\newcommand{\Xtgt}{\mathcal{X}_{\text{t}}}
\newcommand{\Mtgt}{\mathcal{M}_{\text{t}}}
\newcommand{\Xpub}{\mathcal{X}_{\text{p}}}
\newcommand{\Zpub}{\mathcal{Z}_{\text{p}}}

\newtheorem{theorem}{\textbf{Theorem}}
\newtheorem{definition}{\textbf{Definition}}
\newtheorem{lemma}{\textbf{Lemma}}

\DeclareMathOperator*{\argmin}{arg\,min}

\newtheorem{remark}{Remark}

\DOI{foobar}

\cclogo{\includegraphics{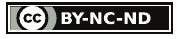}}
  
\begin{document}

\author*[1]{Yue Niu}

\author[2]{Ramy E. Ali}

\author[2]{Salman Avestimehr}

\affil[1]{Electrical and Computer Engineering Department, University of Southern California, yueniu@usc.edu}
\affil[2]{Electrical and Computer Engineering Department, University of Southern California, {reali, avestime}@usc.edu}

  \title{3LegRace: Privacy-Preserving DNN Training over TEEs and GPUs}

  \runningtitle{3LegRace: Privacy-Preserving DNN Training over TEEs and GPUs}


  \begin{abstract}
    {Leveraging parallel hardware (e.g. GPUs) for deep neural network (DNN) training brings high computing performance. However, it raises data privacy concerns as GPUs lack a trusted environment to protect the data. 
    Trusted execution environments (TEEs) have emerged as a promising solution to achieve privacy-preserving learning. Unfortunately, TEEs' limited computing power renders them not comparable to GPUs in performance.
    To improve the trade-off among privacy, computing performance, and model accuracy, we propose an \emph{asymmetric} model decomposition framework, \AsymML{}, to (1) accelerate training using parallel hardware; and (2) achieve a strong privacy guarantee using TEEs and differential privacy (DP) with much less accuracy compromised compared to DP-only methods. 
    By exploiting the low-rank characteristics in training data and intermediate features, \AsymML{} asymmetrically decomposes inputs and intermediate activations into low-rank and residual parts.
    With the decomposed data, the target DNN model is accordingly split into a \emph{trusted} and an \emph{untrusted} part. The trusted part performs computations on low-rank data, with low compute and memory costs. The untrusted part is fed with residuals perturbed by very small noise.  
    Privacy, computing performance, and model accuracy are well managed by respectively delegating the trusted and the untrusted part to TEEs and GPUs. 
    We  provide a formal DP guarantee that demonstrates that, for the same privacy guarantee, combining asymmetric data decomposition and DP requires much smaller noise compared to solely using DP  without decomposition. This improves the privacy-utility trade-off significantly compared to using only DP methods without decomposition.
    Furthermore, we present a rank bound analysis showing that the low-rank structure is preserved after each layer across the entire model. 
    Our extensive evaluations on DNN models show that \AsymML{} delivers $7.6\times$ speedup in training compared to the TEE-only executions while ensuring privacy.
    We also demonstrate that \AsymML{} is effective in protecting data under common attacks such as model inversion and gradient attacks.}
  \end{abstract}

  \journalname{Proceedings on Privacy Enhancing Technologies}
  \DOI{Editor to enter DOI}
  \startpage{1}
  \received{..}
  \revised{..}
  \accepted{..}

  \journalyear{..}
  \journalvolume{..}
  \journalissue{..}

\maketitle
\section{Introduction}\label{sec:intro}
Deep neural networks (DNNs) are acting as an essential building block in various applications such as computer vision (CV) \cite{VGG,ResNet} and natural language processing (NLP) \cite{BERT}. Efficiently training a DNN model usually requires a large training dataset and sufficient computing resources. 
In many real applications, datasets are locally collected and not allowed to be publicly accessible, while training is computation-intensive and hence usually offloaded to parallel hardware (e.g., GPUs). 
Considering that data transfer can be hacked or a runtime memory with sensitive data can be accessed by third parties, such practice poses serious privacy concerns.

\begin{figure*}[t]
    \centering
    \includegraphics[scale=0.3]{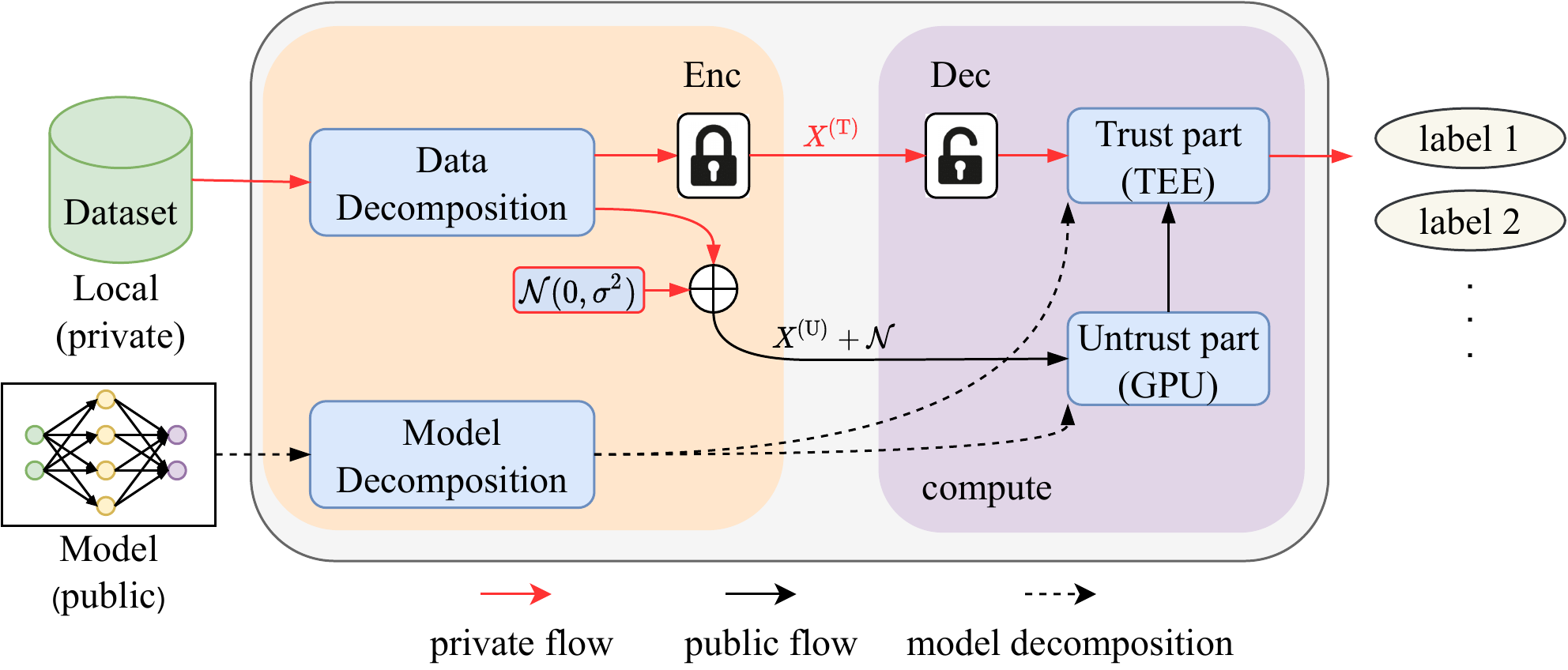}
    \caption{\footnotesize{An Overview of \AsymML{} is depicted. In \AsymML, the models are asymmetrically decomposed into trusted and untrusted parts. The trusted part performs computations on low-rank part $\XT$ at small compute/memory cost; while $\XUT$ perturbed by small noise is offloaded onto the untrusted part with little critical information involved.}}
    \label{fig:AsymML}
    \vspace{-.2cm}
\end{figure*}

The need for private data protection has motivated privacy-preserving machine learning methods such as machine learning with differential privacy (DP) \cite{DiffPrivacy, scalablePATE, DPPerGrad} and machine learning using trusted execution environments (TEEs) \cite{secureTF,MLcapsule,Slalom}. \\ DP-based methods usually defend against membership inference and model-inversion attacks \cite{Membershipattack, ModelInversion} that aim to infer or reconstruct the training data through the DNN models \cite{DiffPrivacy}. Specifically, by injecting noise to the gradients during training, DP-based methods reduce the correlation between the model parameters and the training data. Therefore, reconstructing training data through the model becomes more challenging. 
In addition to applying DP to the models,  \cite{DPData,Shredder} also apply DP directly to the input data during training. 
However, DP alone usually compromises too much accuracy to achieve strong privacy guarantees \cite{DPaccuracy}. \\
Unlike such DP methods, TEEs provide a direct hardware solution to protect data from any untrusted entities. TEE-based methods do not compromise accuracy as in DP. Trusted platforms such as Intel Software Guard Extensions (SGX) \cite{IntelSGX} and Arm TrustZone \cite{ArmTrust} create a sufficiently secure runtime environment, where sensitive data can be stored and processed. By performing all computations in TEEs, any untrusted access to internal memory is forbidden. 
However, the computing performance is severely affected when the TEEs are solely leveraged due to their limited storage and computation capabilities as a result of not having GPU support. 
Such executions are known as the TEE-only executions. A natural solution for this problem is to leverage both TEEs and the untrusted GPUs while protecting the privacy of the data. This idea was leveraged in \cite{Slalom} to develop a privacy-preserving framework known as Slalom for DNN inference. 
However, Slalom does not support DNN training, which is a more challenging and computationally-intensive task.

\textbf{Contributions} -- To efficiently perform private training in a heterogeneous system with \emph{trusted} TEE-enabled CPUs and fast \emph{untrusted} accelerators, we propose a new training framework, \AsymML{}, as shown in Figure~\ref{fig:AsymML}. By exploiting the potential low-rank structure in the inputs and intermediate features $X$, \AsymML{} first \emph{asymmetrically} decomposes the inputs and the features into a low-rank part $\XT$ and  a residual $\XUT$. \AsymML{} then accordingly decomposes the model into trusted and untrusted parts, in which the trusted part is fed with $\XT$, and the untrusted part is fed with $\XUT$ perturbed by a very small noise. 
When the inputs and intermediate features have a low-rank structure, the trusted part with $\XT$ incurs small computation and memory costs, while the untrusted part handles most computations with little privacy revealed.
The DNN model training is performed by respectively delegating the trusted and untrusted part to TEEs and GPUs, where TEEs protect the privacy, and GPUs guarantee the computing performance. We provide a formal DP guarantee that shows that, for the same privacy level, the asymmetric decomposition together with DP requires much smaller noise compared to solely leveraging DP methods. This implies that the decomposition improves the privacy-utility trade-off significantly. 
We also present a theoretical analysis showing that the low-rank structure is preserved after each layer in the DNN model, which ensures efficient asymmetric decomposition. 
In summary, our contributions are as follows.
\begin{enumerate}
    \item We propose a privacy-preserving training framework that decomposes the data, the intermediate features and the models into two parts which decouples the information from the computations. The decomposition is based on a lightweight singular value decomposition (SVD) algorithm along with the Gaussian DP mechanism to protect the privacy further.
    \item We provide essential theoretical analyses that show that a) \AsymML{} achieves strong DP guarantee with a very small noise added compared to the case where there is no decomposition, and b) the low-rank structure in intermediate features is well-preserved after each layer in a DNN model.
    \item We demonstrate \AsymML{} is robust against common machine learning attacks, such as a strong model inversion attack that uses residual data as prior knowledge \cite{SecretRevealer}, and an attack that leverages the  gradients \cite{deepgradients}.
    \item We implement \AsymML{} in a heterogeneous system with TEE-enabled CPUs and fast GPUs that supports various models. \sloppy Our extensive experiments show that AsymML achieves up to $7.6\times$ training speedup in VGG and $5.8\times$ speedup in ResNet variants compared to the TEE-only executions.
\end{enumerate}

\textbf{Organization.} The rest of this paper is organized as follows. In Section \ref{sec:related}, we provide a brief overview of the closely-related works. We then present \AsymML{} and its overhead analysis in Section \ref{sec:asymml}. Our theoretical analyses on privacy and low-rank structure are provided in Section \ref{sec:rankbound}. In Section \ref{sec:attack}, we present several attack models.
In Section \ref{sec:exp}, we provide extensive experiments to evaluate the performance, the overheads and the privacy guarantees of \AsymML. Section \ref{sec:discussion} discusses the potential applications and limitations of \AsymML. Finally, in Section \ref{sec:conclusion}, we discuss some concluding remarks.

\section{Related Works} \label{sec:related}
Ensuring data privacy is  a critical issue in offloading machine learning tasks to distributed and cloud-based systems. Many solutions have been developed to address such an issue, which falls into three main categories: differential privacy \cite{DiffPrivacy}, crypto-based methods \cite{MLconfidential}, and hardware solutions such as TEEs \cite{IntelSGX}.

\textbf{Differential privacy}. Model training using differential privacy (DP) \cite{DiffPrivacy, scalablePATE, DPPerGrad} aims to defend against membership inference \cite{Membershipattack} and model inversion attacks \cite{ModelInversion}. A well-known DNN training algorithm with DP is proposed in \cite{DiffPrivacy}. It follows a typical DP procedure and approximates the training objective function using a Gaussian noise mechanism.  
By injecting noise into the gradients during the training, DP reduces correlation between model parameters and training datasets to some extent. 
Therefore, with such a ``noisy” model, it becomes more challenging to predict if a data record exists in the training dataset (membership inference attacks); or to directly reconstruct the training dataset through the trained model (model inversion attacks). 
However, DP usually greatly compromises accuracy to achieve strong privacy guarantees \cite{DPaccuracy}.  
In addition to applying DP to models, \cite{DPData,Shredder} directly add random noise to the input data, therefore hiding plain data from untrusted parties. However, the accuracy of these methods degrades significantly as we require stronger privacy guarantees. 

\textbf{Crypto-based methods}. Machine learning with encrypted data has been recently investigated using various techniques such as homomorphic encryption (HE) \cite{privMLHEE,CareNet,CryptoNets}, coded computing \cite{privateML} and multi-party computing (MPC) \cite{CrypTen, MPCDL, Goten}. These methods first encrypt the inputs, and then perform the computations in the encrypted domain. 
Despite their effectiveness in preserving privacy, these techniques still face many challenges. 
For example, machine learning with HE incurs much more computation costs compared to normal training flow using plain data. As a result, it is usually impracticable for most current DNN training or inference. 
On the other hand, coded computing \cite{privateML} does not introduce noticeable compute costs, but requires a strong condition that a certain number of compute nodes cannot collude in a distributed system. Such an assumption may not hold in practice. Similar condition is also required in MPC settings. Moreover, current privacy-preserving coded computing approaches are limited to simple machine learning models such as logistic regression \cite{privateML,tang2021verifiable}. 

\textbf{Hardware solutions}. In addition to algorithmic designs, recent works such as \cite{Tensorscone, secureTF, MLcapsule, Citadel, Plinius} have proposed privacy-preserving machine learning methods by leveraging trusted execution environments (TEEs) such as Intel SGX \cite{IntelSGX} and Arm TrustedZone \cite{ArmTrust}. These solutions keep private data in a trusted runtime environment that forbids any untrusted access and then perform  training and inference inside this environment. 
Therefore, they usually offer strong privacy guarantees.
However, TEE-based solutions usually degrade the computing performance due to the lack of GPU support, therefore causing significant training and inference slowdowns.  
To speed up computing, Graviton \cite{Graviton} prototypes a TEE fabric in a GPU platform so that both privacy and computing performance are achieved. On the other hand, \cite{Slalom, DarKnight, Goten} propose solutions that leverage both TEE-enabled CPUs and untrusted GPUs to perform inference. In Slalom \cite{Slalom}, for instance, the inputs to a convolution   layer are first masked with noise in TEEs to ensure privacy and then are sent to GPUs. When the convolutions are done, the noisy outputs are transferred back to TEEs and unmasked using pre-computed results. However, Slalom only supports inference. To further support training in TEEs, methods such as DaKnight \cite{DarKnight}, and Goten \cite{Goten} combine TEE with MPC techniques and require non-colluding compute nodes. As a result, they still fail to protect privacy if such an assumption does not hold.



In addition to the single-user offloading setting considered in our work, federated learning (FL) considers a collaborative learning setting with multiple clients and a central server aiming to learn a global model without data sharing \cite{FederatedLearning}. 
In FL, each client trains a model using the local dataset, while after a certain number of local iterations, participating clients send their local models to a central parameter server. 
By aggregating local models from clients using algorithms such as FedAvg \cite{FedAvg}, the server obtains an updated global model and then broadcasts to clients again in the next round. 
Such a procedure usually repeats many times until convergence is reached.
During training with FL, sensitive datasets collected by each client are always stored in local memory, while the remote server can only access model parameters. 
Therefore, the datasets are protected as long as the local system is trusted. 
However, FL still faces many challenges in practice. 
For example, learning can be very difficult if data are highly heterogeneous \cite{FedNonIID}. 
Furthermore, FL cannot protect data against model inversion attacks. 
Therefore, FL still needs to be combined with methods such as DP \cite{FedDP} or secure aggregation \cite{SecureAggr} to ensure stronger data protection.

\section{AsymML} \label{sec:asymml}
In this section, we present \AsymML. We start with the threat model considered in our work and then describe \AsymML{} in detail, including data and model decomposition in DNNs through a lightweight SVD approximation. Finally, we provide the computation and the memory costs  of \AsymML.

\textbf{Notations} -- We use lowercase letters for scalars and vectors, and uppercase letters for matrices and tensors. $\overline{X}$ denotes a 2D matrix flattened from a multidimensional tensor $X$, while $\left\| \overline{X} \right\|_F$ denotes the Frobenius norm of the matrix $\overline{X}$. $\overline{X}^{*}$ denotes the transpose of $\overline{X}$. $X_i$ denotes $i$-th slice of a tensor $X\in \mathbb{R}^{N \times \Hi \times \Wi}$, $X_{i,:,:}$, while $X_{i,j}$ denotes $(i,j)$-th slice, $X_{i,j,:}$. We use $\circledast$ to denote convolution, and $\cdot$ for matrix multiplication. For a DNN model, given a loss function $\mathcal{L}$, $\nabla_W\mathcal{L}$ denotes gradients of the loss w.r.t parameters $W$. Finally, we use $\log$ to denote the logarithm to the base $2$. 

\subsection{Threat Model}
Based on the capabilities of common TEE platforms such as Intel SGX \cite{IntelSGX}, we consider the following threat model. 1) An adversary may compromise the OS where the TEE is running. However, it cannot breach the TEE environment, 2) the adversary may access hardware disk, runtime stack, memory outside TEEs and communication between trusted and untrusted environments, 3) the adversary may obtain the model parameters and the gradients during training and then analyze underlying relations with the training data, 4) the adversary may obtain data in untrusted environments and infer information in training data, and 5) the adversary may gain some knowledge of the training dataset (e.g., labels), and use public/online resources to improve its attack performance.

\noindent However, we do not consider side-channel attacks that compromise TEEs by probing physical signals such as power consumption and electromagnetic leaks.

\subsection{Asymmetric Data and Model Decomposition using SVD}\label{sec:AsymMLSVD}
\AsymML{} decomposes compute-intensive and memory-intensive modules, especially convolutional layers in modern DNNs \cite{VGG, ResNet}, and then assigns each part to a suitable platform. At a high level, \AsymML{} starts with decomposing a convolutional layer such that the computation involving privacy-sensitive information is performed in TEEs while the residual part is offloaded to GPUs. 
Specifically, the input of a convolutional layer denoted by $X$ is decomposed into a low-rank part $\XT$ and a residual part $\XUT$ as shown in Fig.~\ref{fig:AsymML_Decompose}. 
The residual part $\XUT$ is then protected by adding a small noise as we discuss in Section \ref{sec:rankbound} to ensure privacy. \\ 
When the convolutions in GPUs and TEEs are completed, outputs from GPUs denoted by $\YUT$ are merged into TEEs denoted by $\YT$, and then followed by a non-linear layer (a pooling layer might be also needed). 
Outputs after a non-linear layer will be then re-decomposed before proceeding to the next convolution layer. 
During the whole forward/backward pass, the low-rank part in TEEs is never exposed, which effectively prevents crucial parts of data from being leaked. 
\begin{figure*}
    \centering
    \includegraphics[width=0.6\linewidth]{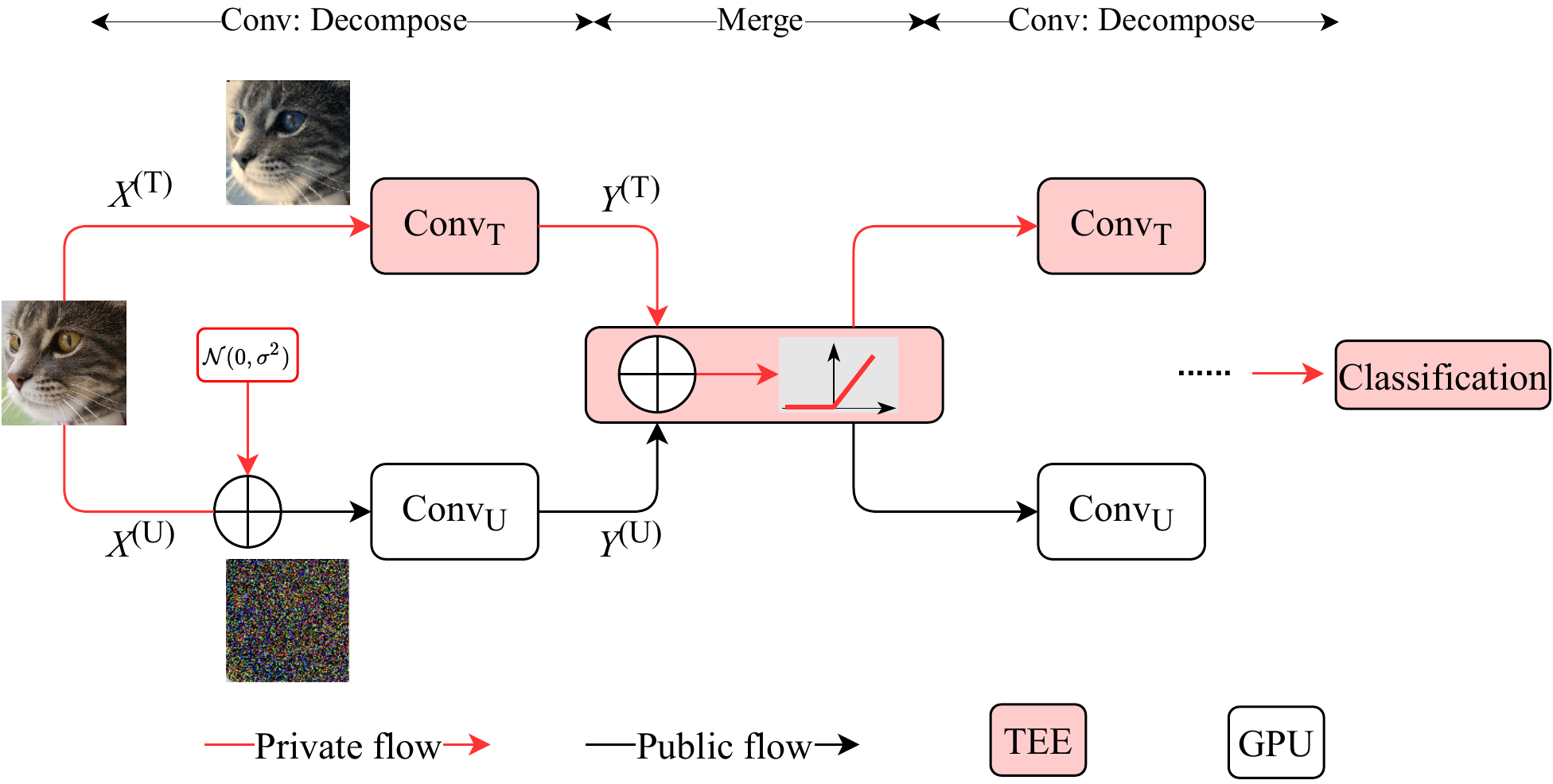}
    \caption{\footnotesize{An illustration of the model decomposition in \AsymML{} is depicted. First, the input of a convolutional layer $X$ is decomposed into a low-rank part $\XT$ and a residual part $\XUT$ that is protected by Gaussian noise. The convolution of the trusted and the untrusted parts are then performed in TEEs and GPUs, respectively. Finally, the results are combined in TEEs followed by a non-linear layer and then the decomposition is performed again and so on. Therefore, training using \AsymML{} behaves like a ``three-legged race".}}
    \vspace{-.4cm}
    \label{fig:AsymML_Decompose}
\end{figure*}

We now explain the decomposition in detail. For a {convolutional layer} with input $X\in \mathbb{R}^{N \times \Hi \times \Wi}$ and kernels $W\in \mathbb{R}^{\co\times\ci\times k\times k}$, where $N$ and $M$ are the number of input and output channels, and $h,w$ and $k$ are the size of inputs and kernels respectively, the $i$-th output channel is computed as follows\footnote{The batch size and the bias are omitted here for simplicity.}
\begin{align}
Y_{i} = \conv(X, W_{i}) = \sum_{j=1}^{N} X_{j} \circledast W_{i,j}. 
\end{align}
By splitting the input $X$ into $\XT$ and $\XUT$, \AsymML{} decomposes the convolution into a trusted convolution $\convtrust$ and an untrusted convolution $\convuntrust$ as follows
\begin{equation}
    \label{eq:conv_split}
    Y_{i} = \convtrust(\XT, W_{i})+\convuntrust(\XUT, W_{i}).
\end{equation}

\noindent The \emph{asymmetric} decomposition is inspired by an observation that channels in inputs and intermediate activations are usually highly correlated. By exploiting such channel correlations, $X$ can be decomposed in a way that a low-rank tensor $\XT$ keeps most information, while $\XUT$ stores the residuals. Therefore, the complexity of $\convtrust$ is significantly reduced with little privacy compromised. \\ 
To extract a low-rank tensor $\XT$, we first apply singular value decomposition (SVD) to $\overline{X} \in \mathbb{R}^{N \times \Hi\Wi}$ flattened from $X$ as
\begin{equation}
    \label{eq:X_svd}
    \overline{X} = U \cdot \text{diag}(s) \cdot V^{*},
\end{equation}
where the \emph{principal} channels are stored in $V$ while the corresponding singular values in $s$ denote the importance of each principal channel. The $j$-th channel in $\XT$ is obtained from the first $\cii$ most principal channels as follows 
\begin{equation}
    \label{eq:XT}
    \XT_{j} = \sum_{\jj=1}^{\cii}s_{\jj} \cdot U(j,\jj) \cdot X_{\jj}^{'} \equiv \sum_{\jj=1}^{\cii}a_{j,\jj} \cdot X_{\jj}^{'},
\end{equation}
where $X_{\jj}^{'} \in \mathbb{R}^{\Hi\times\Wi}$ is a 2D matrix reshaped from the $\jj$-th column of $V$ and $a_{j,\jj}=s_{\jj} \cdot U(j,\jj)$.
On the other hand, the  residual part $\XUT_j$ is given by
\begin{align}
\XUT_j = X_j - \XT_j = \sum_{\jj=\cii+1}^{N}a_{j,\jj} \cdot X_{\jj}^{'}.
\end{align}

\noindent To further strengthen the privacy guarantee, $\XUT$ is perturbed using a very small Gaussian noise as 
\begin{align}
    \mathcal{M}(\XUT) = \XUT + Z
\end{align}
before being sent to the GPUs, where $Z_{i, j}$ are independent $\mathcal{N}(0,\sigma^2)$ random variables. 

\noindent With the low-rank input $\XT$, the forward and backward passes of $\convtrust$ can be reformulated in a way such that the complexity depends on the number of principal channels $\cii$ as follows. 

\begin{itemize}
    \item 
 \textbf{Forward.} During a forward pass, the outputs in TEEs are calculated as follows 
\begin{equation}
    \label{eq:YT}
    \begin{split}
    \YT_{i} &= \sum_{j=1}^{\ci}\XT_{j} \circledast W_{i,j} =\sum_{j=1}^{\ci}\sum_{\jj=1}^{\cii}a_{j,\jj}X^{'}_{\jj} \circledast W_{i,j}\\
    &=\sum_{\jj=1}^{\cii}X^{'}_{\jj} \circledast \sum_{j=1}^{\ci}a_{j,\jj}W_{i,j} \equiv \sum_{\jj=1}^{\cii}X^{'}_{\jj} \circledast W_{i,\jj}^{'},
\end{split}
\end{equation}
where $W_{i,\jj}^{'}=\sum_{j=1}^N a_{j,\jj}W_{i,j}$.
According to Eq.~(\ref{eq:YT}), $\convtrust$ essentially is a convolution operation with $\cii$ input channels, and the kernels $W^{'}$ that are obtained by regrouping $W$. 

\item \textbf{Backward.} During a backward pass, given gradient $\nabla_Y\mathcal{L}$, $\dWT$ in TEEs is computed as
\begin{equation}
    \label{eq:dW}
    \begin{split}
        \dWT&=\XT_{j}\circledast\nabla_{Y_i}\mathcal{L} =\sum_{\jj=1}^{\cii} a_{j, \jj}\cdot X^{'}_{\jj}\circledast\nabla_{Y_i}\mathcal{L},
    \end{split}
\end{equation}
where $X^{'}_{\jj}\circledast\nabla_{Y_i}\mathcal{L}$ for $\jj=1,2,\cdots,\cii$ are first computed. $\dWT$ is obtained by a simple linear transformation with $a_{j,\jj}$. Hence, the computation complexity of backward passes also depends on $\cii$. \\
On the other hand, the untrusted forward output $\YUT_i$ and the backward gradient $\dWUT$  in GPUs are the same as in the classical convolutional layers but with a different input $\mathcal{M}(\XUT)$. The final result $Y_i$ and $\dW$ are obtained by simply adding the results of the TEEs and the GPUs. 
In addition, as computing $\dX$ does not involve the input $X$, it is offloaded onto GPUs.
\end{itemize}
\subsection{Lightweight SVD Approximation}
\label{subsec:lightweight}
Performing exact SVD on $X\in \mathbb{R}^{N \times \Hi \times \Wi}$ incurs a significant complexity of $O(N\Hi^2\Wi^2+N^2\Hi\Wi)$, which goes against our objective of reducing the complexity in TEEs. 
To reduce the complexity, we propose a lightweight SVD approximation that reduces complexity to only $O(\cii N\Hi\Wi)$.

SVD essentially is an algorithm that finds two vectors $\ui{i}$ and $\vi{i}$ to minimize $\left \| \Xin{i-1}-\ui{i}\cdot{\vi{i}}^* \right \|_F$, where $\Xin{i-1}$ is the remaining $\overline{X}$ with $i-1$ most principal components extracted.
$\left \{ \ui{i} \right \}_{i=1}^{\Hi\Wi}$ and $\left \{ \vi{i} \right \}_{i=1}^{N}$ are both orthogonal set of vectors. \AsymML{}, however, does not require such orthogonality. Therefore, SVD can be then relaxed as  
\begin{equation}
\label{eq:X_svd_light}
    \begin{aligned}
        & \XTfl = \argmin_{\XTfl} \quad \left \| \overline{X} - \XTfl \right \|_F^2, \\ 
        & \textrm{s.t.} \ \text{rank} \left(\XTfl \right)\leq \cii,\quad \XUfl = \overline{X} - \XTfl.
    \end{aligned}
\end{equation}

\noindent Using alternating optimization \cite{AlterOpt}, each component $i$ in $\XTfl$ can be obtained as described in Algorithm~\ref{alg::lightSVD}\footnote{In the actual implementation, $\XTfl$ is stored as a list of vectors $\left\{ \ui{i}, \vi{i} \quad | \quad i=1,\cdots,\cii \right \}$, rather than as a matrix.}. 
Due to the fast convergence of alternating optimization, given suitable initial values (e.g. output channels from the previous ReLU/Pooling layer), we have experimentally observed that the maximum number of iterations $\mathrm{max\_iter}$ to reach a near-optimal solution $\left \{ \ui{i} \right \}_{i=1}^{\cii}$ and $\left \{ \vi{i} \right \}_{i=1}^{\cii}$ is typically $1\sim 2$ (See Appendix \ref{appx:approxsvd}). 
Finally, it is worth noting that the computation complexity of this algorithm is much less than exact SVD as it only increases linearly with $\cii$.

\begin{algorithm}
    \SetAlgoLined
    \KwData{$\cii, \overline{X}, \left\{ \ui{i}_0, \vi{i}_0 \quad | \quad i=1,\cdots,\cii \right \}, \mathrm{max\_iter}$}
    \KwResult{$\XTfl, \XUfl$}
    Initialize $\XTfl$ as $0$\;
    \For{$i$ \textbf{in} $1,\cdots,\cii$}{
        \For{$j$ \textbf{in} $1,\cdots, \mathrm{max\_iter}$}{
            \tcc{Alternating optimization}
            $\ui{i}_{j} = \frac{\overline{X}\cdot \vi{i}_{j-1}}{\left \| \vi{i}_{j-1}\right \|_F^2}$\;
            $\vi{i}_{j} = \frac{\overline{X}^*\cdot \ui{i}_{j}}{\left \| \ui{i}_{j}\right \|_F^2}$\;
        }
    $\XTfl=\XTfl + \ui{i}_j \cdot {\vi{i}_j}^*$\;
    $\overline{X} = \overline{X} - \ui{i}_j\cdot{\vi{i}_j}^*$\;
    }
    $\XUfl = \overline{X}$\;
\caption{Lightweight SVD Approx.}
\label{alg::lightSVD}
\end{algorithm}


\subsection{Overhead Analysis}
The computation and memory costs in TEEs are of great concern as they decide \AsymML{}'s performance in a heterogeneous system. In this section, we break down these costs and compare the costs in TEEs and GPUs.

\begin{figure*}
    \centering
    \begin{subfigure}{0.45\textwidth}
        \centering
        \includegraphics[width=0.92\linewidth]{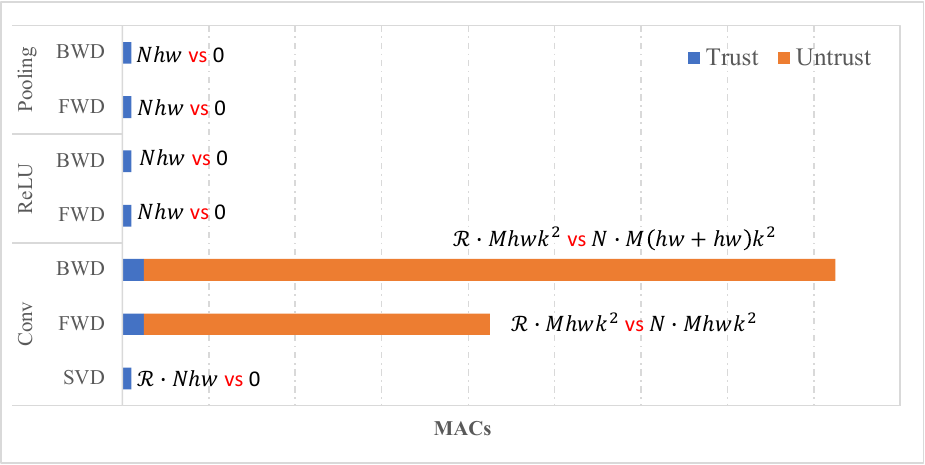}
        \caption{Computation in TEEs and GPUs}
        \label{fig:complexity_comp}
    \end{subfigure}%
    \begin{subfigure}{0.45\textwidth}
        \centering
        \includegraphics[width=0.95\linewidth]{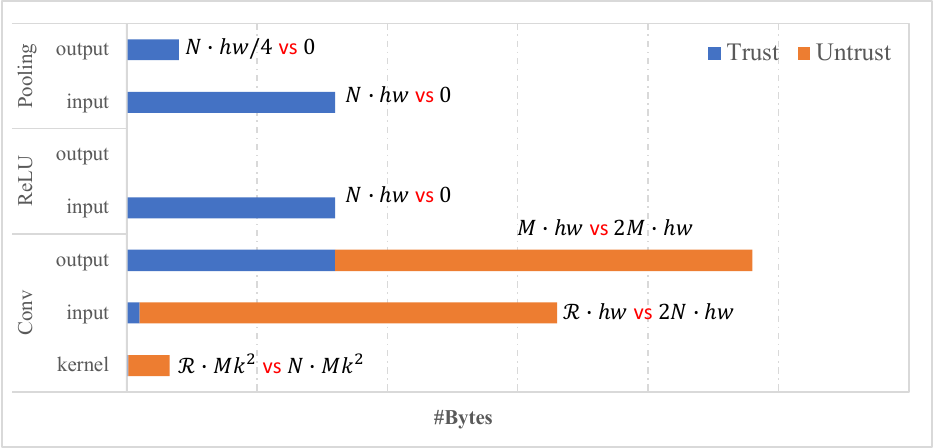}
        \caption{Memory in TEEs and GPUs}
        \label{fig:complexity_mem}
    \end{subfigure}
    \caption{\footnotesize{The computation and the memory costs in TEEs (Trusted) and GPUs (Untrusted) ($\frac{\cii}{\ci}=\frac{1}{16}$, pooling kernel size is $2$)}}
    \vspace{-.2cm}
    \label{fig:complexity}
\end{figure*}

Figure~\ref{fig:complexity} shows computation and memory costs in TEEs and GPUs for the case where $\cii/N=1/16$. 
As described in Section \ref{sec:AsymMLSVD}, the computation complexity of $\convtrust$ in TEEs increases with the number of principal channels $\cii$, while the complexity of $\convuntrust$ in GPUs is the same as that of the original convolutional layer. 
In addition to $\convtrust$, all element-wise operations (ReLU, Pooling, etc) and the lightweight SVD are performed in TEEs. 
As for the memory cost, the inputs and outputs of all element-wise operations are stored in TEEs, together with the inputs $\XT$ and the outputs of the convolution $\convtrust$,  $\YT$. 
On the other hand, inputs $\mathcal{M}(\XUT)$, outputs $\YUT$ of $\convuntrust$ and the corresponding gradients $\dX$ and $\dY$ are stored in GPUs.

\section{Theoretical Guarantees}\label{sec:rankbound}
In this section, we first analyze the low-rank structure of the intermediate activations in NNs, and show that such a low-rank structure is preserved after linear operators such as convolutions in convolutional neural networks (CNNs).
Then, we characterize the differential privacy guarantee of \AsymML{} when the input $\XUT$ is perturbed by a small Gaussian noise $Z$. We further demonstrate that \AsymML{} requires much smaller noise for the privacy budget compared to directly adding noise to original inputs as in \cite{Shredder}.


\subsection{Low-Rank Structure in NNs}
\label{subsec:rankbound}
We first define a metric termed as \emph{SVD-channel entropy} to formally quantify the low-rank structure in the intermediate features. Then, we show how the SVD-channel entropy changes in a DNN model after each layer. Inspired by SVD entropy \cite{SVDEntropy} and Rényi entropy \cite{RenyiEntropy}, we define SVD-channel entropy based on singular values obtained in  (\ref{eq:X_svd}). 
Given the singular values $\left \{ s_{\jj} (X) | \ \jj=1,\cdots, \ci\right \}$ of an input matrix $X$ with $N$ channels, the SVD-channel entropy of $X$, denoted by $\mu_X$, is defined as follows.
\begin{definition}(SVD-Channel Entropy).
\label{def:CEntropy}
The SVD-channel entropy of a matrix $X$ is given by
\begin{equation}
\mu_X = -\log \left(\sum\limits_{j=1}^{\ci} \bar{s}_j^2(X) \right),
\end{equation}
where $\bar{s}_j(X)=\frac{s_j(X)}{\sum\limits_{\jj=1}^{\ci}s_{\jj}(X)}$ is the $j$-th normalized singular value.
\end{definition}

\noindent Next, we show in Lemma~\ref{lemma:CEntropy} and Theorem \ref{theorem::CEntropy} that $\mu_X$ defined above indicates the number of \emph{principal} channels actually needed to approximately reconstruct the original data $X$.
\begin{lemma}
\label{lemma:CEntropy}
The SVD-channel entropy of an input $X$ with $N$ input channels is bounded as $0\leq \mu_X \leq \log \ci$. 
\end{lemma}

\noindent If we use the first $\ceiling{2^{\mu_X}}$ most principal channels to reconstruct $X$, then Theorem~\ref{theorem::CEntropy} shows that such a reconstruction is sufficient to approximate $X$.
\begin{theorem}
\label{theorem::CEntropy}
Given a matrix $X$ with SVD-channel entropy $\mu_X$, and assuming the $j$-th singular value is given as $s_j(X)=a\cdot b^{j-1}$, for some constants $a>0$, $0<b<1$,
if we use the $\cii=\ceiling{2^{\mu_X}}$ most principal channels to reconstruct $X$, then we have $\frac{\sum_{j=1}^{\cii}s_j^2(X)}{\sum_{j=1}^{\ci}s_j^2(X)} \geq 0.97$.
\end{theorem}
\begin{remark} \normalfont 
The assumption in Theorem~\ref{theorem::CEntropy} that $s_j(X)=a\cdot b^{j-1}$ usually hold in natural images, where the data have highly correlated channels. In such cases, the singular values usually decay exponentially \cite{IQA, TraceNorm}. 
\end{remark}

\noindent The proofs of Lemma~\ref{lemma:CEntropy} and Theorem~\ref{theorem::CEntropy} are provided in Appendix A. In the rest of the paper, we regard $\cii$ as the sufficient number of  channels to represent $X$.

\textbf{SVD-Channel Entropy in CNNs} -- In DNNs, given input data with low-rank structure, it is crucial to measure how such structure (measured using SVD-channel entropy) changes after every layer so that we can systematically set the number of principal channels in TEEs. 
In CNNs, convolutional, batch normalization, pooling and non-linear layers such as ReLU are the basic layers. 
In this section, we mainly analyze how these basic operators change the structure of the data. All proofs are provided in Appendix A.

\noindent We now start with the convolutional layers. \\
\textbf{Convolutional layer} -- For a convolutional layer with input $X \in R^{\ci\times\Hi\times\Wi}$, we first bound the SVD-channel entropy of the outputs with $1\times 1$ kernels in Theorem \ref{theorem:CEntropy_conv1x1} and then we extend this to the general case of $k\times k$ kernels in Theorem \ref{theorem:CEntropy_convkxk}.
\begin{theorem}
\label{theorem:CEntropy_conv1x1}
For a convolution layer, given input $X \in R^{\ci\times\Hi\times\Wi}$ with SVD-channel entropy $\mu_X$, kernel $W \in R^{\co\times\ci\times 1\times 1}$, then the SVD-channel entropy of the output $Y\in R^{\co\times\Ho\times\Wo}$ is upper-bounded as follows 
\begin{center}
    $\mu_Y \leq \log \ceiling{2^{\mu_X}}$.
\end{center}
\end{theorem}

\noindent Theorem~\ref{theorem:CEntropy_conv1x1} implies that low-rank structure is still preserved in the outputs for  convolutional layers with $1\times 1$ kernels. Therefore, before a convolutional layer, if $\ceiling{2^{\mu}}$ principal channels are used in TEEs, the same number of channels is still sufficient to approximate the outputs.

\begin{figure}
    \centering
    \includegraphics[width=.7\linewidth]{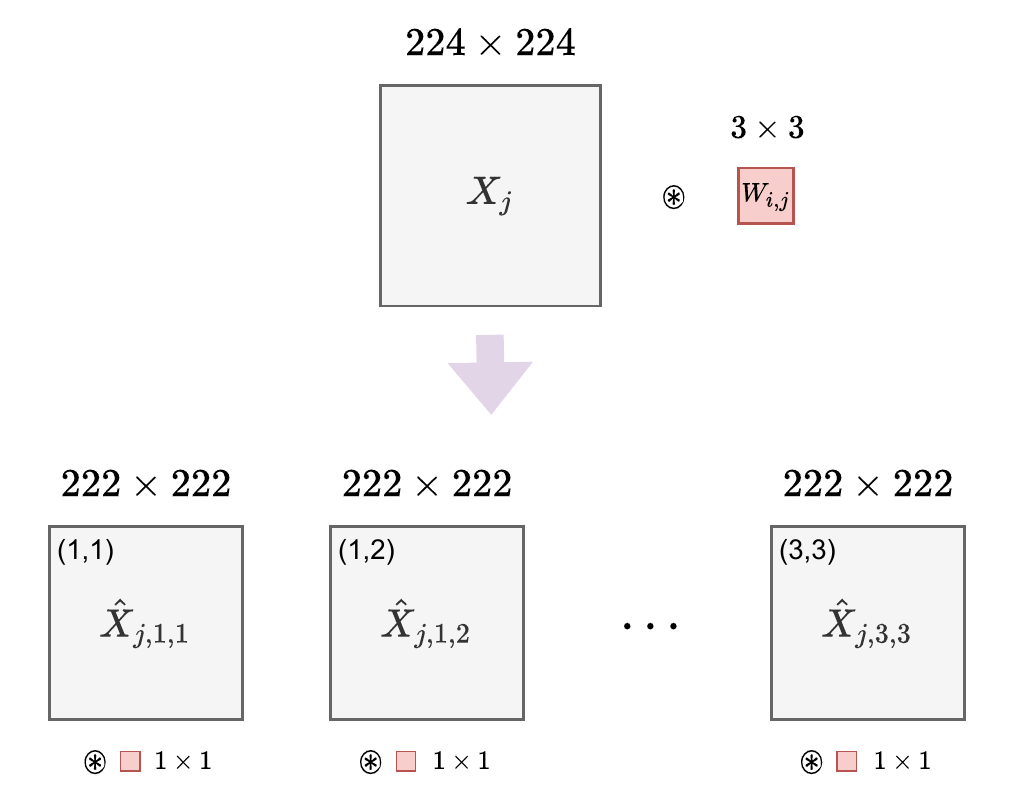}
    \caption{\footnotesize{An illustration of the conversion from $k\times k$ convolution to $1\times 1$ convolutions with $k^2$ different input channels.}}
    \vspace{-.4cm}
    \label{fig::k_to_1_conv}
\end{figure}

For a convolutional layer with a $k\times k$ kernel, it can be viewed as a $1\times 1$ convolutional layer with $k^2$ different ``input channels". Each ``input channel" is a patch of $X_j$, denoted as $\left \{ \hat{X}_{j,q,r} | 1\leq q, r\leq k \right \}$, as shown in Figure~\ref{fig::k_to_1_conv}. 
Given $X \in R^{\ci\times\Hi\times\Wi}$, $W \in R^{\co\times\ci\times k\times k}$, the $i$-th output channel can be rewritten as follows 
\begin{equation}
\label{eq::conv_k_to_1}
    Y_{i} = \sum_{j=1}^{\ci} W_{i,j}\circledast X_j=\sum_{j=1}^{\ci} \sum_{q=1,r=1}^{k,k} W_{i, j, q, r} \circledast \hat{X}_{j,q,r},
\end{equation}
where $W_{i,j,q,r}$ is a $1\times 1$ kernel after this conversion. \\
Therefore, a $k\times k$ convolution with $\ci$ input channels is equivalent to a $1\times 1$ convolution with $\ci\cdot k^2$ ``input channels": $\left \{ \hat{X}_{j,q,r} | 1\leq j\leq\ci,1\leq q, r\leq k \right \}$. To simplify the notation, we denote $\hat{X}_{:,q,r}$ as the $(q,r)$-th patch of all channels, and $\hat{X}_{j,:,:}$ as all patches of the $j$-th channel. 

We now show in Theorem~\ref{theorem:CEntropy_convkxk} that, knowing the SVD-channel entropy of $\hat{X}_{j,:,:}$ for $j=1,\cdots, \ci$, the SVD-channel entropy of the outputs with $k\times k$ convolution can be accordingly bounded.
\begin{theorem}
\label{theorem:CEntropy_convkxk}
Given data $X$ with SVD-channel entropy $\mu_X$, $\cii = \ceiling{2^{\mu_{X}}}$,   $\mu_{\hat{X}_{j,:,:}} = \hat{\mu}_j$ for $1\leq j \leq \ci$ and WLOG suppose that $\hat{\mu}_1 \leq \hat{\mu}_2 \leq \cdots \leq \hat{\mu}_N$, then the SVD-channel entropy of $\hat{X}$ satisfies
\begin{align*}
\mu_{\hat{X}} \leq \log (\sum_{j=1}^{\cii} \ceiling{2^{\hat{\mu}_j}}) \cong \mu_X + \bar{\mu},
\end{align*}
where $\bar{\mu} = \frac{\sum_{j=1}^{\cii} \hat{\mu}_j}{\cii}$. Furthermore, the SVD-channel entropy of the output $Y$ after convolution with $W \in R^{\co\times\ci\times k\times k}$ is upper-bounded as 
\begin{center}
    $\mu_Y \leq  \log (\sum_{j=1}^{\cii} \ceiling{2^{\hat{\mu}_j}})$
\end{center}
\end{theorem}
Based on Theorem \ref{theorem:CEntropy_convkxk}, with $k\times k$ kernels, the SVD-channel entropy of the outputs increases with $\mu_{\hat{X}_{j,:,:}}$ for $j=1,\cdots,\cii$. Intuitively, if $k$ is larger, more patches are included, then $\mu_{\hat{X}_{j,:,:}}$ will be higher. 
More numerical analyses are presented in Section~\ref{subsec:entropykern}. 

\textbf{BatchNorm layer} -- Batch normalization is commonly used in CV models to reduce the \emph{internal covariate shift}. According to Theorem~\ref{theorem:CEntropy_batchnorm}, the SVD-channel entropy of outputs is almost the same as that of inputs in a BatchNorm layer.

\begin{theorem}
\label{theorem:CEntropy_batchnorm}
Given data $X$ with SVD-channel entropy $\mu_X$, then the SVD-channel entropy of output $Y$ after a batch normalization layer is upper-bounded as
\begin{center}
    $\mu_Y\leq \log(\ceiling{2^{\mu_X}} + 1)$.\\
\end{center}
\end{theorem}
Hence, the low-rank structure of inputs is still preserved after batch normalization.

\textbf{ReLU layer} -- As a non-linear layer, ReLU is the most commonly used operator in  DNNs. Theoretically bounding the  SVD-channel entropy after ReLU is infeasible. 
Instead, we empirically measure the SVD-channel entropy after a ReLU in Appendix~\ref{sec:SEntropyModels} and show that the ReLU layers also preserve the low-rank structure. 

\textbf{Pooling layer} -- Pooling operations include max and average pooling.  Max pooling is a also non-linear operator. Similarly, we provide an empirical experiment that shows that a max pooling layer does not greatly change the low-rank structure either (See Appendix~\ref{sec:SEntropyModels}). Besides, the max and the average pooling layers usually deliver similar performance. As stated in Theorem~\ref{theorem:CEntropy_pooling}, the SVD-channel entropy after an average pooling layer is less than the one before pooling. Therefore, pooling layers still preserve the low-rank structure.
\begin{theorem}
\label{theorem:CEntropy_pooling}
For an average pooling layer, given input $X \in R^{\ci\times\Hi\times\Wi}$ with SVD-channel entropy $\mu_X$, the SVD-channel entropy in output $Y$ satisfies
\begin{align*}
    \mu_Y \leq \log \left \lceil 2^{\mu_X} \right \rceil.
\end{align*}
\end{theorem}

\subsection{Privacy Guarantees}
We now provide the DP guarantee of \AsymML{}. \AsymML{} adds a small Gaussian noise to $\XUT$ to ensure privacy. While the data stored in TEEs are protected by a secure hardware, the Gaussian mechanism further protects the information in the GPUs.

First, for a dataset $\mathcal{X} = \left \{ X^1, X^2, \cdots, X^i, \cdots, X^B \right \}$, we define the neighboring dataset $\mathcal{X}^{'} = \left \{ X^1, X^2, \cdots, 0, \cdots, X^B\right \}$ for any possible $i$-th record removed \cite{DiffPrivacy}. 
During training, we sample a batch $\mathcal{S}\in\mathcal{X}$ with size $b$. The  $\ell_2$-sensitivity of the lightweight SVD is given by $\Delta_2 =\sup_{\mathcal{X},\mathcal{X}^{'}}\norm{\mathcal{X}^{(U)}-\mathcal{X}^{'(U)}}\leq \xi\sup_{i}\norm{X^i}$, where $\xi$ is the upper bound of the ratio between the Frobenius norms of $\XUT$ and $X$.
Theorem \ref{theorem:dp} states that \AsymML{} is $(\epsilon, \delta)$-differentially private when a Gaussian noise $\mathcal N(0, 2\Delta^2_2 \ln{(1.25q/\delta)}/\epsilon^2. I)$ is added, where $q=b/B$ is the sampling probability and $I$ is the identity matrix. 

\begin{theorem}\label{theorem:dp}
 \AsymML{} is $(\epsilon, \delta)$-differentially private when a Gaussian noise $\mathcal N(0, 2\Delta^2_2 \ln{(1.25q/\delta)}/\epsilon^2. I)$ is added to the residuals, for $0<\epsilon, 0<\delta\leq q$, where $b$ is the batch size, $q=b/B$ is the sampling probability and $\Delta_2 \leq \xi\sup_{i}\norm{X^i}$ with SVD.
\end{theorem}
\begin{remark}

Compared to directly adding noise to the original data $X$,
\AsymML{} reduces the required noise variance by $\xi^2$. 
Therefore, \AsymML{} significantly improves the utility-privacy trade-off  (See also the experiments of Section \ref{subsec:accuracy}).
\end{remark}

\section{Attack Models}\label{sec:attack}
Attack models are a crucial part in evaluating a privacy-preserving NN training/inference algorithm. An attacker usually uses any relevant information available such as model parameters, gradients, and any public prior information to reveal the private information being protected. 
In NNs, attacks reconstructing training data using the model parameters or the gradients are one of the strongest attack methods. The reason for the success of such attacks arises from leveraging the correlations between training data and model parameters as well as gradients during the training.

In this section, we consider two attacks that aim to reconstruct the target training dataset. The first attack is a model inversion (MI) attack \cite{SecretRevealer} that uses \emph{well-trained} models and some prior knowledge (noisy residual data $\DUT$ in our case) to find synthetic images similar to images in the target training dataset. 
The second attack is a gradient inversion attack \cite{deepgradients}. It finds synthetic data that can result in similar gradients as the original training dataset.
While a model inversion attack is mainly targeting a well-trained model, a gradient inversion attack can happen at any stage, especially for an initial model or a pre-trained model \cite{deepgradients}.
In addition, it is worth noting that common membership inference attacks that leverage output probabilities \cite{Membershipattack} fail in \AsymML{} since the logits are secured in TEEs.

\textbf{Notations}. We use $\Xtgt$ to denote the target training dataset, and, $\Zpub$ to denote the prior knowledge known to the attacker, including some public dataset $\Xpub$ and the residuals $\DUT$ available in untrusted platforms. Note that $\Xpub$ does not overlap with $\Xtgt$, but it may contain objects with similar labels.
Finally, $\Mtgt$ represents the target model under attack.

\subsection{Model Inversion Attack}
We now consider the model inversion attack.

\textbf{Assumptions} -- We assume that the attacker has access to the target model $\Mtgt$. Further, the attacker uses the noisy residuals $\DUT$ as prior information to help reconstruct the training dataset $\Xtgt$. In addition, the attacker knows general information such as label information. Therefore, it can find relevant resources (e.g. online images) to learn the target data distribution. 
\begin{figure*}[ht!]
\centering
\captionsetup{justification=centering}
\begin{subfigure}{.48\linewidth}
    \centering
    \includegraphics[width=.78\linewidth]{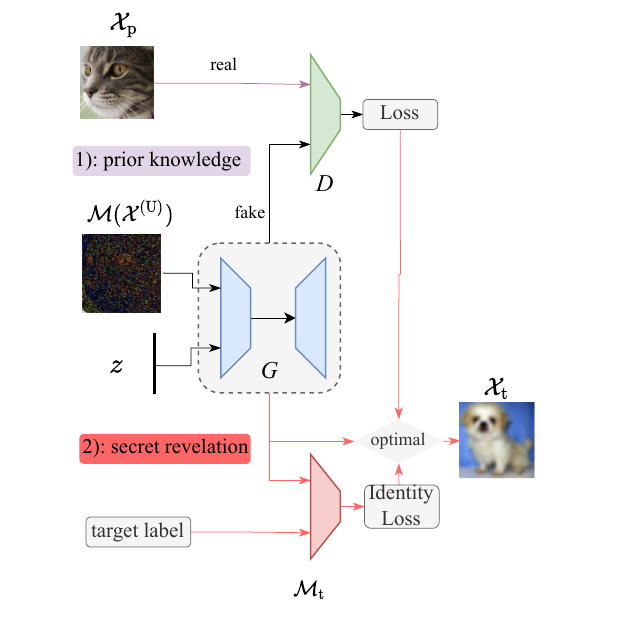}
    \caption{An illustration of the model inversion attack using the residual data $\XUT$ and any public relevant data as prior information. The attacker first trains a GAN model using prior knowledge (e.g., $\DUT$). Then, with the residual data $\DUT$ from the target training dataset along with the labels, the attacker uses the generator $G$ to reconstruct images that achieve the highest accuracy in the target model $\Mtgt$}
    \label{fig:attackmodel}
\end{subfigure}
\hfill
\begin{subfigure}{.48\linewidth}
    \centering
    \includegraphics[width=.9\linewidth]{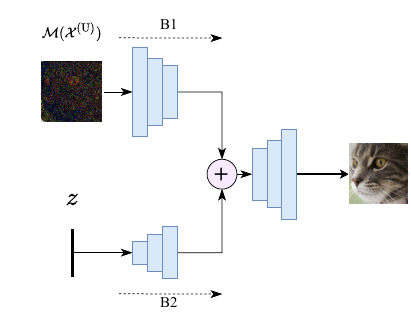}
    \caption{An illustration of the generator architecture that uses residual data $\DUT$ to reconstruct images. The generator consists of two branches: B1 for extracting features from residual data $\DUT$, and B2 for generating latent features. Then, these two branches are merged and further up-sampled using deconvolution to generate outputs that have similar distribution as the training dataset.}
    \label{fig:generator}
\end{subfigure}
\caption{Model inversion attack using $\DUT$ as prior knowledge.}
\vspace{-.4cm}
\end{figure*}

We design the attacker as shown in Fig.~\ref{fig:attackmodel} and Fig. \ref{fig:generator}. Specifically, the attack has two stages as follows.
\begin{enumerate}
    \item \textbf{Prior Knowledge Distillation}. In this stage, the attacker trains a generative adversarial network (GAN) with the public prior knowledge $\Zpub$. 
    
    Fig.~\ref{fig:generator} shows the generator $G$ architecture. It consists of two branches. The first branch B1 is used for learning features from the residual noisy data $\DUT$. The second branch B2 is used for generating latent features. On the other hand, the discriminator $D$ is a classical CNN. The detailed architectures are provided in Appendix~\ref{sec:miArchitecture}. The Wasserstein-GAN loss function is used during training as follows
    \begin{align}
    \min_{G}\max_{D} L(G, D) &= E_{x\in\Xpub}[D(x)]- \notag \\&E_z[D(G(z,\DUT))].
    \end{align}
    \item \textbf{Secret Revelation}. In this stage, with the noisy residual data $\DUT$ from the target dataset $\Xtgt$ along with the labels, the attacker uses the generator to reconstruct images that achieve the highest accuracy in the target model $\Mtgt$.
    
    We find the optimal latent vector $z$ that generates an almost-real image to the discriminator, while achieving the lowest identity loss in the target model $\mathcal{M}_{\text{t}}$, namely
    \begin{align}
        \hat{z}=\arg\min_{z} L_{\text{prior}}+\lambda L_{\text{id}},
    \end{align}
    where $L_{\text{prior}}$ is the loss in the discriminator, $L_{\text{id}}$ is the loss in the target model, and $\lambda$ controls the weight between $L_{\text{prior}}$ and $L_{\text{id}}$. These losses are given by
    \begin{equation}
        L_{\text{prior}} = -D(G(z)), \quad L_{\text{id}} = -\log [\mathcal{M}_t(G(z))],
    \end{equation}
    where $\mathcal{M}_t(G(z))$ is the output probability from the target model.
\end{enumerate}

\subsection{Gradient Inversion Attack}
Next, we consider a gradient version attack.

\textbf{Assumptions} -- We assume the attacker has access to the gradients in the untrusted GPUs, $\dWUT$, as well as the gradients on the inputs $\dX$. Similar to the model inversion attack, the attacker can use $\DUT$ as prior knowledge to help reconstruct synthetic images with similar gradients as the target dataset.

Figure \ref{fig:gradinversion} shows the procedures used in conducting a gradient inversion attack. Knowing the gradients $\nabla_W^{(U)}\mathcal{L}$ and $\dX$ from target images, the attack first feeds $\Mtgt$ with the noisy residual data $\DUT$.
After obtaining the gradients $\nabla_W^{(U)}\mathcal{L}^{'}, \nabla_X\mathcal{L}^{'}$, the $\ell_2$ distance between the gradients is computed as follows
\begin{equation}\label{eq:gradattackloss}
    L = \norm{\nabla_W^{(U)}\mathcal{L}-\nabla_W^{(U)}\mathcal{L}^{'}}^2 +\lambda\norm{\dX-\nabla_X\mathcal{L}^{'}}^2.
\end{equation}
Based on $L$, the attacker computes the gradients on $\mathcal{M}(\XUT)$, and optimizes inputs using a common gradient descent method.
\begin{figure}
    \centering
    \includegraphics[width=.9\linewidth]{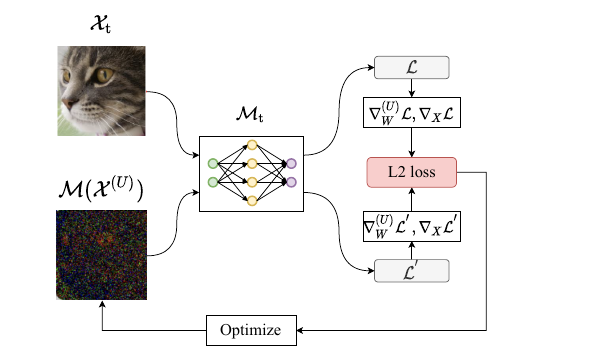}
    \caption{An illustration of the gradient inversion attack that generates synthetic images with similar gradients as the target dataset.}
    \label{fig:gradinversion}
\end{figure}

\section{Empirical Evaluation}\label{sec:exp}
In this section, we evaluate \AsymML{} in terms of the training accuracy, running time, robustness against attacks and information leakage. We perform our experiments on the following models and datasets. \\
For models, we consider VGG-16, VGG-19 \cite{VGG}, ResNet-18 and ResNet-34 \cite{ResNet}.
For datasets, we consider the CIFAR-10 \cite{CIFAR10} and the ImageNet \cite{ImageNet} datasets.  

\noindent We illustrate the implementation of \AsymML{} in detail in Section \ref{subsec:implementation}. Then, in Section \ref{subsec:entropykern},  we study the effect of the kernel size on the SVD-channel entropy. In Section \ref{subsec:accuracy}, we study the training accuracy of \AsymML{}. In Section \ref{subsec:running-time}, we compare between \AsymML{} and the baselines in terms of the running time. Finally, we present the privacy guarantee of \AsymML{} and its capability under 
model inversion and gradient inversion attacks in Section \ref{subsec:attackperf}.

\subsection{\AsymML{} Implementation}
\label{subsec:implementation}
\begin{figure*}[htb!]
    \centering
    \includegraphics[width=.7\linewidth]{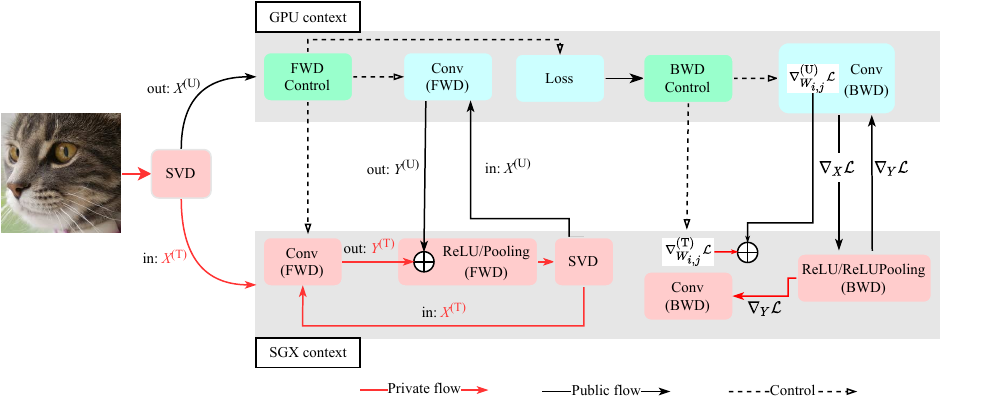}
    \caption{\footnotesize{Implementation of \AsymML{}  in a heterogeneous system with SGX enabled CPUs and GPUs. The forward and the backward passes in convolutional layers are coordinated by FWD and BWD Control, respectively. }}
    \vspace{-.4cm}
    \label{fig:asymml_impl}
\end{figure*}

\noindent We implement \AsymML{} in a heterogeneous system as shown in Figure \ref{fig:asymml_impl} with an Intel SGX enabled Xeon CPUs \cite{IntelSGX} and NVIDIA RTX5000 GPUs working as untrusted accelerators. \AsymML{} first reads a predefined model (e.g. from PyTorch \cite{PyTorch}), decomposes it into trusted and untrusted parts, and then offloads them to TEEs and GPUs respectively.
The operations in TEEs (e.g. $\convtrust$, ReLU, and Pooling) are supported by dedicated trusted functions. In order to reduce the CPU-GPU communications, a Pooling is fused with the preceding ReLU layer. 
Similar optimization also is applied to  convolutional layers and the following batch normalization layers.

During training, SGX and GPU contexts are created for the trusted and the untrusted operations. We use PyTorch as a high-level coordinator to distribute the computations, activate GPU/SGX context, compute loss, and update the model parameters (See the forward and backward control in Figure~\ref{fig:asymml_impl}). 

During a forward pass, outputs of a convolutional layer in GPUs and TEEs will be merged in the following ReLU (and Pooling) layer in TEEs.
A lightweight SVD is then applied to decompose the output activations into low-rank and residual parts, which are then fed into the next convolution layer. 

During a backward pass, computing $\dX$ for ReLU and Pooling layers is performed in TEEs, while computing $\dX$ for convolutional layers is performed in GPUs. The partial gradients $\dWT$ and $\dWUT$ in a convolutional layer are computed in GPUs and TEEs respectively and merged into TEEs.

\subsection{SVD-Channel Entropy vs Kernel Size}\label{subsec:entropykern}
In this subsection, we investigate the effect of the kernel size on the SVD-channel entropy.
As Theorem~\ref{theorem:CEntropy_convkxk} states, the upper bound of the SVD-channel entropy of outputs in a convolutional layer increases with the kernel size. 
More patches are generated with large kernels in each channel (See Fig.~\ref{fig::k_to_1_conv}), therefore this leads to increasing SVD-channel entropy along the patches ($\mu_{\hat{X}_{j,:,:}}$). 

We present an empirical result on the SVD-channel entropy $\mu_{\hat{X}_{j,:,:}}$ with kernel sizes ranging from $1$ to $11$. 
In this experiment, the input data $X$ are randomly sampled from ImageNet. As detailed in Section \ref{subsec:rankbound}, for $k\times k$ kernels, $k^2$ input patches are generated. 
We compute the SVD-channel entropy along all patches. 
Fig. \ref{fig:CEntropy_patches} shows the SVD-channel entropy across $10$K randomly sampled images with different kernel size. The line plot indicates the mean value, while the violin plots show the distribution. 
With $3\times 3$ kernels, the SVD-channel entropy increases by around $1$, which implies that given $\cii$ principal channels in inputs, $2 \cii$ principal channels are sufficient for outputs after a convolution layer with $3\times 3$ kernel. 
Although it is noted that the SVD-channel entropy increases with the kernel size, small kernel sizes are commonly used in modern DNNs such as VGG and ResNet. Therefore, the SVD-channel entropy increment is well-managed.

\subsection{Training Accuracy}\label{subsec:accuracy}
Next, we investigate the training performance. In the experiments, we follow standard data preprocessing procedure. The training hyperparameters are as follows: weight decay is $0.0005$, momentum is $0.9$, the maximum number of epochs is $200$ for CIFAR-10 and $100$ for ImageNet. The initial learning rate ($\mathrm{lr}$) is $0.1$. We use the cosine annealing learning rate decay strategy.\\
We consider  the following four schemes.
\begin{figure}[!htb]
    \centering
    \includegraphics[width=.8\linewidth]{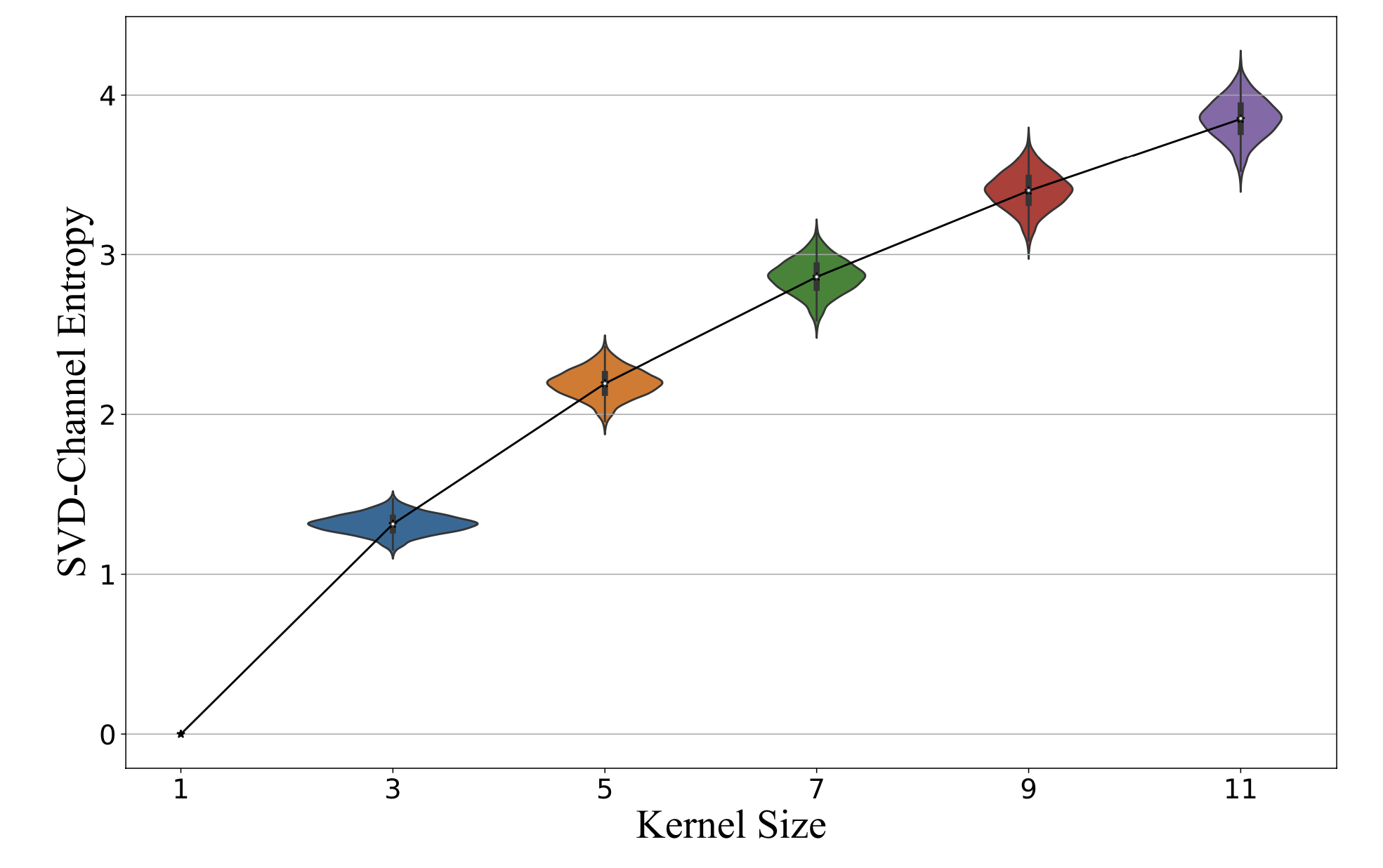}
    \caption{\footnotesize The SVD-channel entropy along patches of a images is shown. The line plot shows the mean value, while the violin plot shows the distribution. A wide violin indicates large number of images are around a particular SVD-channel entropy level.}
    \label{fig:CEntropy_patches}
\end{figure}

\begin{figure}
    \centering
    \includegraphics[width=.95\linewidth]{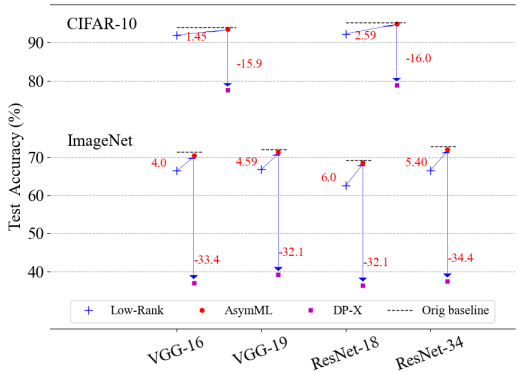}
    \caption{The accuracy after training with only the low-rank part and original data is shown for CIFAR-10 and ImageNet. The low-rank training achieves good accuracy, but not is sufficient compared to the original models. Hence, the residual data $\XUT$ is an indispensable component to filling this gap.}
    \label{fig:lowrankacc}
    \vspace{-.4cm}
\end{figure}
\begin{enumerate}
    \item \textbf{Training with low-rank data $\XT$(Low-Rank)}. In Low-Rank, each convolutional layer is only fed with the low-rank part $\XT$. Other operations are the same as the original model. Based on the theoretical analysis in Section \ref{subsec:rankbound} and the empirical evidence in Section~\ref{subsec:entropykern}, we compute $\XT$ as follows: for the first convolutional layer, $\XT$ only consists of the most principle channel ($\cii=1$); $\cii$ is doubled when another convolutional layer (VGG) or a residual block (ResNet) is added, while $\mathcal R$ remains unchanged for a ReLU or a pooling layer.
    \item \textbf{Training with noise added to the original data $X$ (DP-X).} In DP-X and \AsymML{}, we set $(\epsilon,\delta)=(1,10^{-5})$. According to Theorem \ref{theorem:dp}, we generate noise with the parameters given in Table \ref{tab:noiseparam}.
    \item \textbf{Training with noise added to the residuals $\XUT$ (\AsymML{})}. In \AsymML{}, we add noise to the residual parts obtained by the light-weight SVD. 
    \item \textbf{Training with the original data $X$ (Orig).} In Orig, the training is performed with the original data without SVD and without adding any noise. This results in the best accuracy, but that does not provide any privacy guarantee. 
\end{enumerate}
\noindent Fig.~\ref{fig:lowrankacc} shows the accuracy of Low-Rank, DP-X, \AsymML{} and Orig on CIFAR-10 and ImageNet. Since CIFAR-10 is a very small dataset, we only train smaller models (VGG-16, ResNet-18) on it. 
On CIFAR-10, Low-Rank already achieves very high accuracy. 
\AsymML{} further improves the accuracy by $1.45 \%$ and $2.59 \%$ in VGG-16 and ResNet-18, respectively. The final accuracy is almost the same as the original models. 
However, to achieve the same privacy guarantee, DP-X that directly adds noise to inputs suffers a significant accuracy drop.

Similar results are also observed on ImageNet. While Low-Rank training incurs a slightly larger accuracy drop, \AsymML{} achieves almost the same accuracy as the original models. DP-X still fails to preserve accuracy under the same privacy budget.

\begin{table}[!htb]
\caption{Noise parameters for training in DP-X and \AsymML{}.}
\label{tab:noiseparam}
\centering
\begin{tabular}{c|cccccc}
\toprule
 & b & B & q & $\xi$ &\multicolumn{2}{c}{$\sigma$} \\
\midrule
& & & & & \AsymML{} & DP-X \\
CIFAR-10 & 32 & 50K & 6e-4 & 0.05& 0.12 & 2.5 \\
ImageNet & 128 & 1.2M & 1e-4 & 0.05 & 0.11 & 2.2 \\
\bottomrule
\end{tabular}
\end{table}



\subsection{Runtime Analysis}
\label{subsec:running-time}
We conduct training and inference on ImageNet and compare the runtime with two baselines: GPU-only, TEE-only. Following the standard data pre-processing procedure, we resize each image into $3\times 224\times 224$. The batch size is set to $32$. 
We compute the runtime as the average time of a forward and backward pass.
Computing the number of principal channels $\cii$ is TEEs is similar as in Section \ref{subsec:accuracy}.


\textbf{Training Runtime. } For training, we compare the runtime performance with the TEE-only and the GPU-only executions. Fig.~\ref{fig:training_ImageNet} shows the actual runtime (red plots) and the relative slowdowns (bar plots) compared to the GPU-only execution. 
Compared to the TEE-only method, \AsymML{} achieves $7.5 - 7.6\times$ speedup on VGG-16/VGG-19, and up to $5.8 - 5.8\times$ speedup on ResNet-18/34. By encoding most information in TEEs with very low-rank representations, \AsymML{} achieves significant performance gains compared to TEE-only executions.
Although \AsymML{} shows around $15\times$ slowdown compared to \emph{untrusted} GPU-only execution, \AsymML{} protects the privacy of the training datasets unlike the GPU-only execution. 

\begin{figure}[!htb]
    \centering
    \begin{subfigure}{.5\textwidth}
      \centering
      \includegraphics[width=\linewidth]{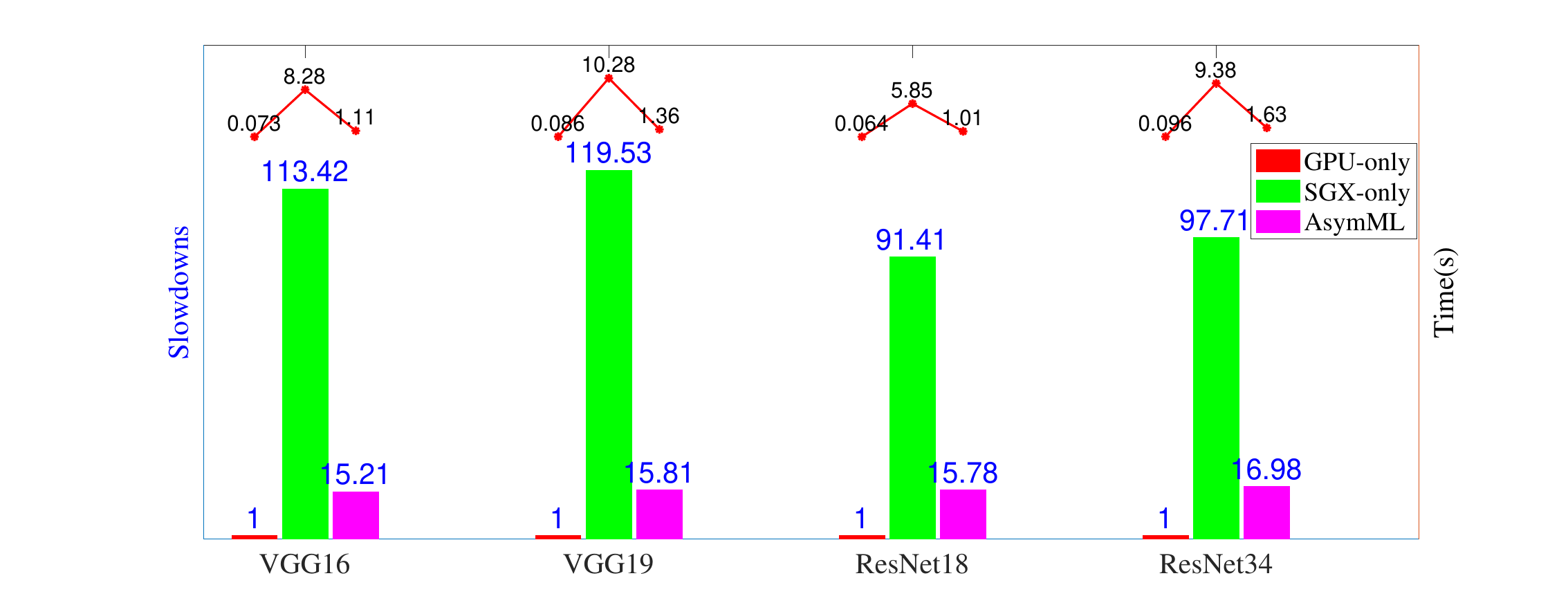}
      \caption{\footnotesize{The training time of \AsymML{} compared to the SGX-only  and the GPU-only executions is shown.}}
      \label{fig:training_ImageNet}
    \end{subfigure}
    \vfill
    \begin{subfigure}{.5\textwidth}
      \centering
      \includegraphics[width=\linewidth]{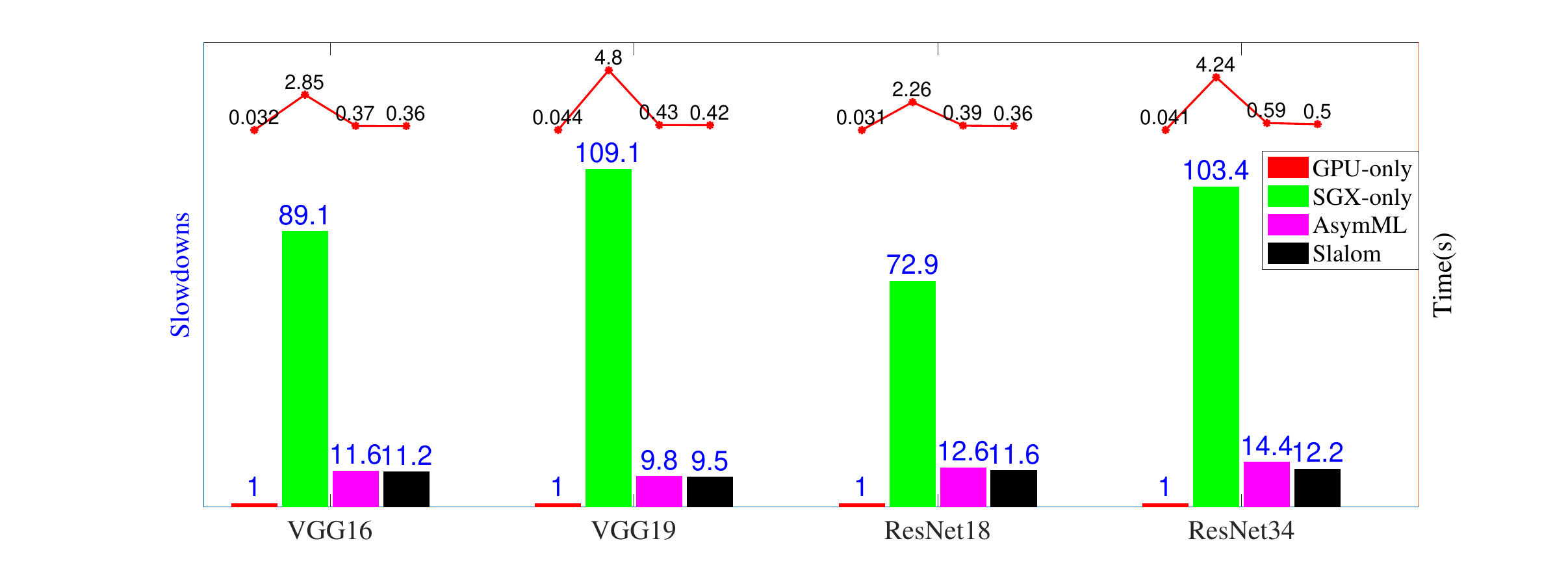}
      \caption{\footnotesize{The inference time of \AsymML{} compared to the SGX-only execution, the GPU-only execution and Slalom is shown.}} 
      \label{fig:inference_ImageNet}
    \end{subfigure}
    \caption{\footnotesize{Training and Inference performance on ImageNet. Bar plot shows slowdowns compared to GPU-only execution; red dotted lines are the corresponding running time.}}
    \label{fig:training_inference_ImageNet}
    \vspace{-.4cm}
\end{figure}

\textbf{Inference Runtime. } We also compare \AsymML{} to Slalom \cite{Slalom} in terms of inference time. While a part of the convolution is performed in TEEs, inference using \AsymML{} is just slightly slower than Slalom. Therefore, the low-rank representation in TEEs does not incur significant additional costs.

\begin{figure}[!htb]
\centering
\begin{subfigure}{\linewidth}
    \centering
    \captionsetup{justification=centering}
    \includegraphics[width=.8\linewidth]{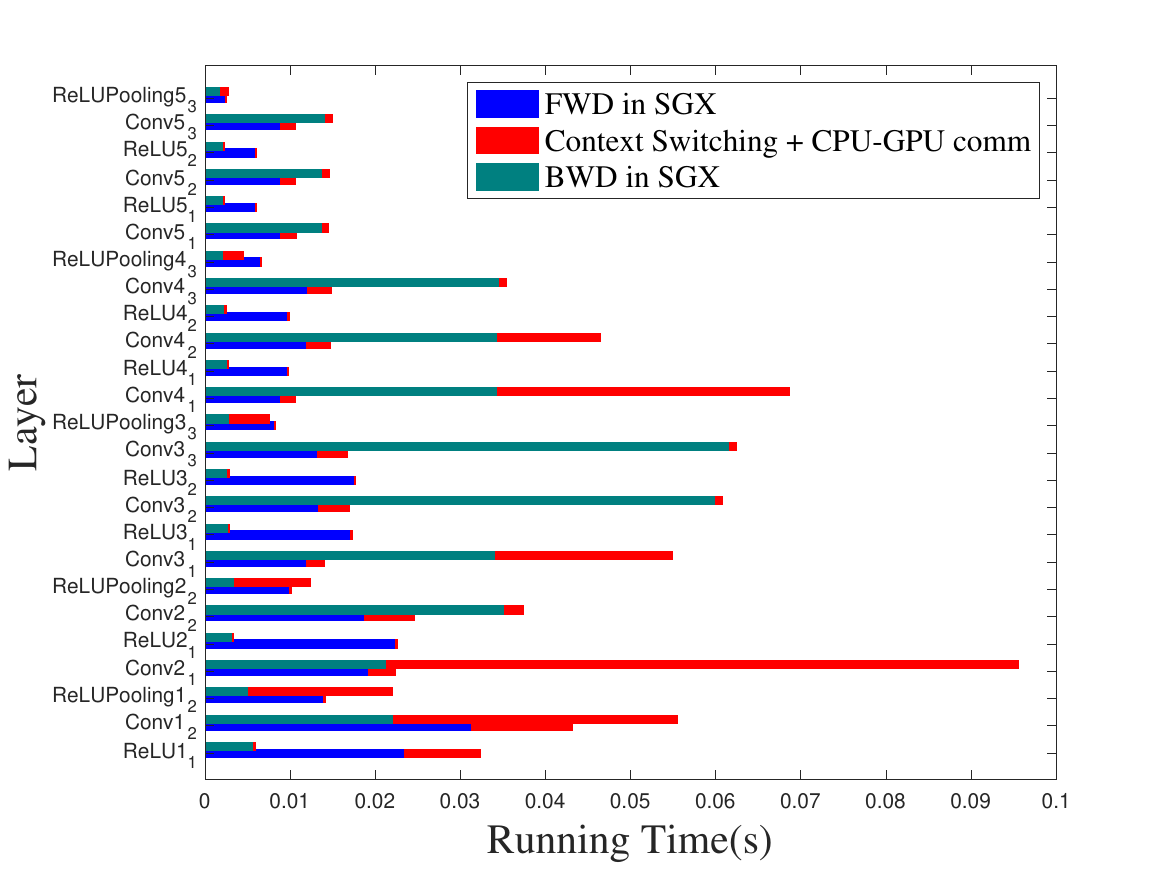}
    \caption{\footnotesize{The runtime of VGG-16 is shown.}}
    \label{fig:runbreakdown_resnet18}
\end{subfigure}
\vfill
\begin{subfigure}{\linewidth}
    \centering
    \captionsetup{justification=centering}
    \includegraphics[width=.9\linewidth]{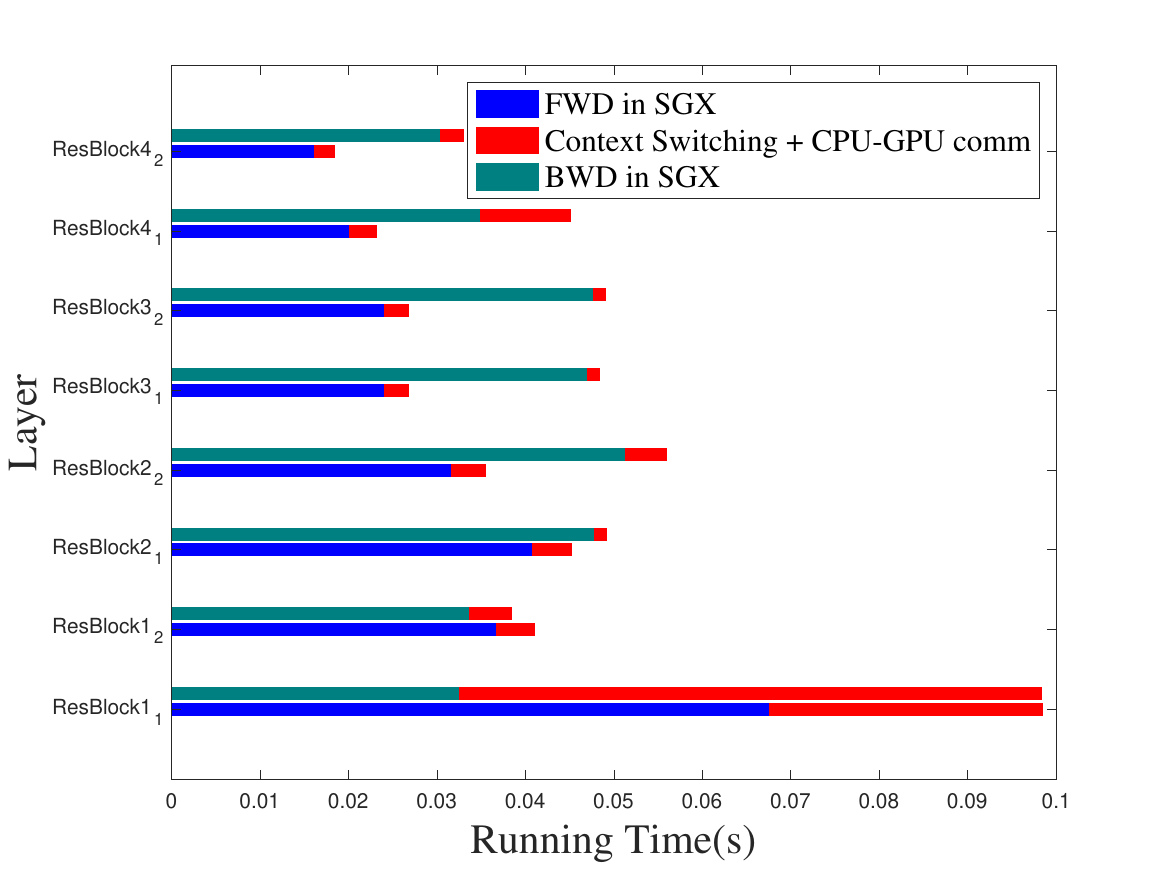}
    \caption{\footnotesize{The runtime of ResNet-18 is shown.}}
    \label{fig:runbreakdown_resnet34}
\end{subfigure}
\caption{An illustration of the runtime breakdown in VGG-16 and ResNet-18 is shown. }
\label{fig:runbreakdown}
\vspace{-.4cm}
\end{figure}

\textbf{Runtime Breakdown. } To better identify the main bottlenecks on a heterogeneous platform, we  profile the running time of the forward and backward passes in \AsymML{}. We break down the running time into computation time and communication time. We show the runtime breakdown for VGG-16 and ResNet-18 in  Fig.~\ref{fig:runbreakdown}. More results are provided in Appendix \ref{sec:runtimemore}.

\begin{itemize}
\item \textbf{Forward Pass}. For the forward pass, the blue bars show the time spent in TEEs, while the red bars show CPU-GPU communication. 
In the early convolutional layers, due to the large number of features, the communication from GPUs to CPUs brings notable costs, which then becomes marginal in later layers. 

\item \textbf{Backward Pass.} For the backward pass, the green bars show the time in TEEs, and the red bars show CPU-GPU communication. The additional cost for CPU-GPU communication is much more dominant, especially in early convolution layers. The reason is that during the backward pass, not only are input gradients $\dX$ but also parameter gradients $\dW$ transferred between CPUs and GPUs. 
\end{itemize}
\subsection{Privacy Protection}
\label{subsec:attackperf}


In this subsection, we evaluate \AsymML{}'s privacy protection against the two attacks of Section \ref{sec:attack}: the model inversion attack and the gradient inversion attack. We use CIFAR-$10$ to demonstrate the performance of \AsymML{} under such attacks in all experiments. We use the same noise in Table \ref{tab:noiseparam}.
We first list the metrics used to evaluate the performance of \AsymML{} and then we discuss detailed results.

\textbf{Metrics.} We consider the following  metrics: peak signal-to-noise ratio (PSNR), structural similarity index (SSIM), and the accuracy under the target model ($\mathrm{Acc}_{\mathcal{M}}$). PSNR is used to measure the pixel-wise similarity between two samples, while SSIM is used to measure the visual quality of a human visual system \cite{psnrssim}. They are both widely used in image quality assessment. Large PSNR ($\leq \infty$) and SSIM ($\leq 1$) indicate more similarity between original and the reconstructed images.
Finally, we also report accuracy on the target model $\mathcal{M}_{\text{t}}$ to test how it recognizes the reconstructed data.

\noindent \textbf{Model Inversion Attack.} In this experiment, we first partition the dataset into two parts. The first part is used as a public dataset $\mathcal{X}_{\text{p}}$ to train the attack model and the other part $\mathcal{X}_{\text{t}}$ is used as a  private dataset to train the target model $\mathcal{M}_{\text{t}}$. 
With $\mathcal{X}_{\text{p}}$ and  $\mathcal{X}_{\text{t}}$ created from one dataset, we simulate the scenario that the attack model can learn general knowledge about the target dataset such as features relevant to the labels. 
We use $\mathcal{M}(\XUT)$ in the first layer ($\XUT$ perturbed with small noise) as the prior information to the attack model as it reveals most information in inputs. 

In the first stage of training the attack model, we train the GAN model using an Adam optimizer with weight decay $0.0005$ and momentum $0.5/0.999$ for both the generator $G$ and the discriminator $D$. The batch size is set as $64$. We set the learning rate as $0.0025$ for $G$ and as $0.01$ for $D$ to achieve the best training performance. In the second stage, we optimize the latent vector $z$ using an SGD optimizer with learning rate $0.01$. Each batch of latent vectors is optimized for $200$ iterations.
\begin{table}[]
\centering
\begin{tabular}{ccccc}
\toprule
 Dataset & Target Model & PSNR & SSIM & $\mathrm{Acc}_{\mathcal{M}}$ \\
\midrule
 CIFAR-10 & VGG-16  & 9.43 & 0.12 & $10.74 \%$ \\
 CIFAR-10 & ResNet-18  & 9.31 & 0.09 & $10.56 \%$ \\
\bottomrule
\end{tabular}
\caption{Similarity between the reconstructed images and the original ones using various metrics in model inversion attacks, where $\text{PSNR}_{\text{max}}=\infty$, and $\text{SSIM}_{\text{max}}=1.0$ for two identical images.}
\label{tab:attackperf}
\end{table}

\begin{figure}[!htb]
\begin{subfigure}{.8\linewidth}
    \centering
    \captionsetup{justification=centering}
    \includegraphics[width=.8\linewidth]{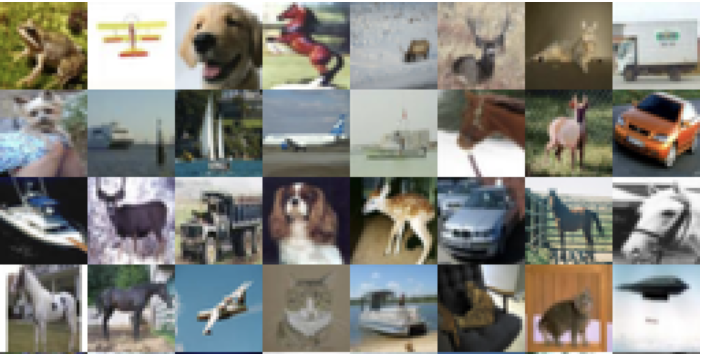}
    \caption{\footnotesize{The original data samples of CIFAR-10 dataset.}}
    \label{fig:attack_orig}
\end{subfigure}
\vfill
\begin{subfigure}{.8\linewidth}
    \centering
    \captionsetup{justification=centering}
    \includegraphics[width=.8\linewidth]{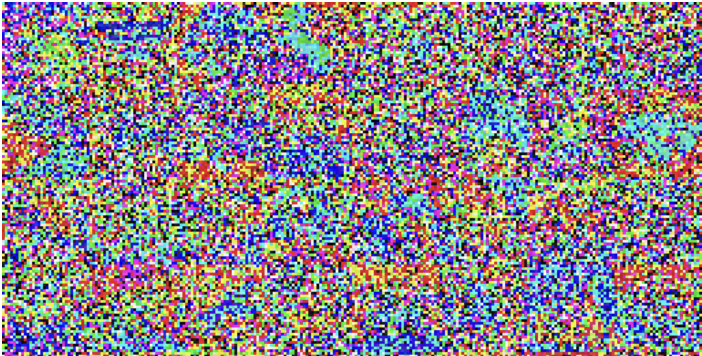}
    \caption{\footnotesize{The noisy residual data $\mathcal{M}(\XUT)$ of CIFAR-10 dataset.}}
    \label{fig:attack_residual}
\end{subfigure}
\begin{subfigure}{.8\linewidth}
    \centering
    \captionsetup{justification=centering}
    \includegraphics[width=.8\linewidth]{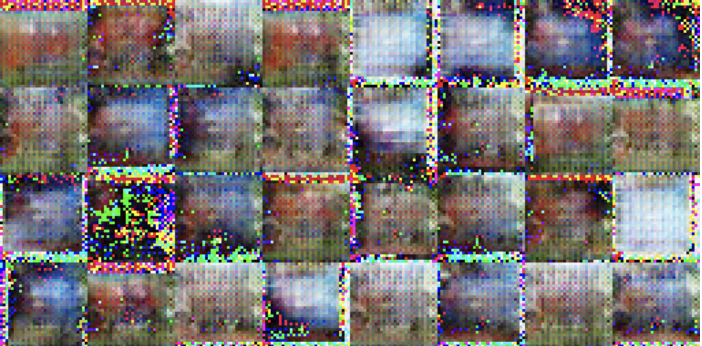}
    \caption{\footnotesize{The reconstructed samples using $\DUT$ in \AsymML{}.}}
    \label{fig:attack_reconasymml}
\end{subfigure}
\vfill
\begin{subfigure}{.8\linewidth}
    \centering
    \captionsetup{justification=centering}
    \includegraphics[width=.8\linewidth]{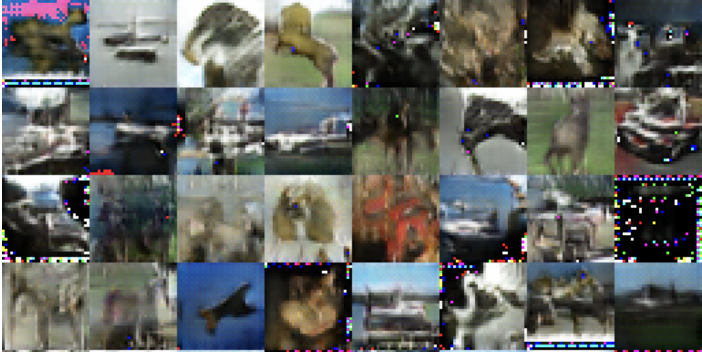}
    \caption{\footnotesize{The reconstructed samples using $\mathcal{X}^{(U)}$.}}
    \label{fig:attack_recon}
\end{subfigure}
\caption{The reconstructed data samples using the model inversion attack with residual data $\mathcal{M}(\XUT)$ as prior information, and ResNet-18 as the target model are shown.}
\label{fig:attackperf}
\end{figure}

Table~\ref{tab:attackperf} lists the performance of the model inversion attack using different metrics. Based on PSNR and SSIM, the reconstructed images using $\DUT$ as prior knowledge share very few similarities with the original ones. Furthermore, as reported in $\mathrm{Acc}_{\mathcal{M}}$ for target models VGG-16 and ResNet-18, with the $\DUT$ as prior information, the target models still fail to match the reconstructed images with their true labels (low $\mathrm{Acc}_{\mathcal{M}}$). 
Fig.~\ref{fig:attackperf} further presents some reconstructed data samples (Fig.~\ref{fig:attack_reconasymml}) compared to the original ones (Fig.~\ref{fig:attack_orig}). The target model is ResNet-18. 
Based on Fig.~\ref{fig:attack_residual}, it is observed that the residual $\DUT$ reveals little feature information about a specific class. As a result, even though it is used as prior knowledge when training GAN models in the MI attack, the reconstructed and original images still share few common features. Therefore, \AsymML{} can effectively defend against this MI attack. 

We further conduct a model inversion reconstruction using $\mathcal{X}^{(U)}$ as prior knowledge, as shown in Fig. \ref{fig:attack_recon}. 
It is interesting to observe that $\mathcal{X}^{(U)}$ without noise does provide a little useful information that helps reconstruct slightly better synthetic images. Therefore, adding small noise to $\mathcal{X}^{(U)}$ as in \AsymML{} is very critical to further hide such information.

\textbf{Gradient Inversion Attack.} In the experiments, we use $\Mtgt$ with random initialized values as in \cite{deepgradients}. $\mathcal{X}^{(U)}$ of the inputs (perturbed by small noise) is used as initial synthetic images.

During optimization, we use an L-BFGS optimizer with an initial learning rate of $0.1$. The history vector size in L-BFGS is $100$. The batch size is $64$. The maximum number of iterations for one optimization is $200$. $\lambda$ is Eq (\ref{eq:gradattackloss}) is $0.5$. Fig \ref{fig:dgattackperf} shows the original images and the reconstructed ones using this gradient inversion attack. It it observed that given $\mathcal{M}(\XUT)$ as initial value (Fig \ref{fig:dgattack_residual}), the gradient attack fails to generate similar images as the original ones (Fig \ref{fig:dgattack_reconasymml}), even though the optimization loss is very small (Fig. \ref{fig:dgattack_loss}). 
It is worth noting that the original gradient attack method \cite{deepgradients} achieves good reconstruction only on a single image, rather than on large batches. Such an observation is also aligned with our results.

We also use $\mathcal{X}^{(U)}$ without noise to perform a gradient inversion attack. As shown in Fig \ref{fig:dgattack_residual}, the reconstructed images are still very noisy, and share very few features with original ones.
\begin{figure}[!htb]
\begin{subfigure}{.8\linewidth}
    \centering
    \captionsetup{justification=centering}
    \includegraphics[width=.8\linewidth]{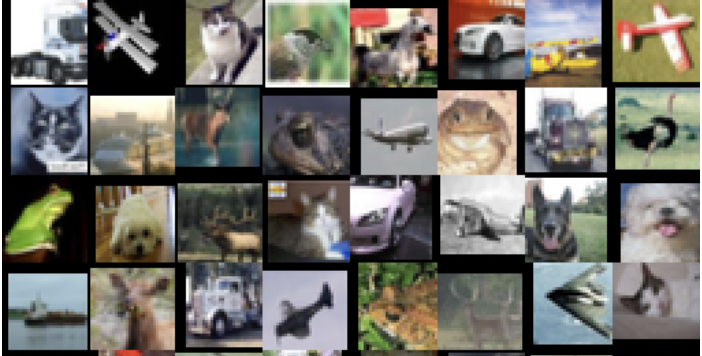}
    \caption{\footnotesize{The original data samples of CIFAR-10 dataset.}}
    \label{fig:dgattack_orig}
\end{subfigure}
\vfill
\begin{subfigure}{.8\linewidth}
    \centering
    \captionsetup{justification=centering}
    \includegraphics[width=.8\linewidth]{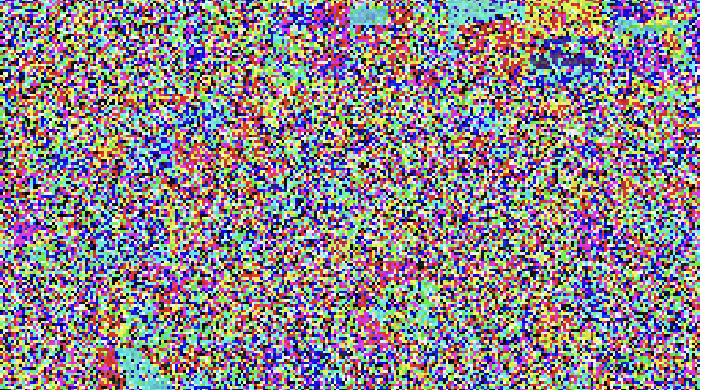}
    \caption{\footnotesize{The residual data $\mathcal{M}(\XUT)$ of CIFAR-10 dataset.}}
    \label{fig:dgattack_residual}
\end{subfigure}
\begin{subfigure}{.8\linewidth}
    \centering
    \captionsetup{justification=centering}
    \includegraphics[width=.8\linewidth]{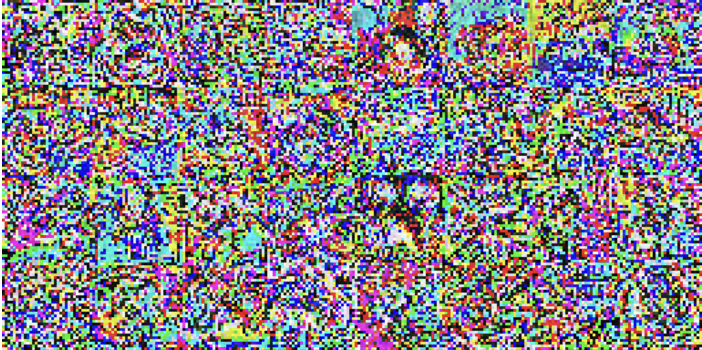}
    \caption{\footnotesize{The reconstructed samples using $\DUT$ in \AsymML{}.}}
    \label{fig:dgattack_reconasymml}
\end{subfigure}
\vfill
\begin{subfigure}{.8\linewidth}
    \centering
    \captionsetup{justification=centering}
    \includegraphics[width=.8\linewidth]{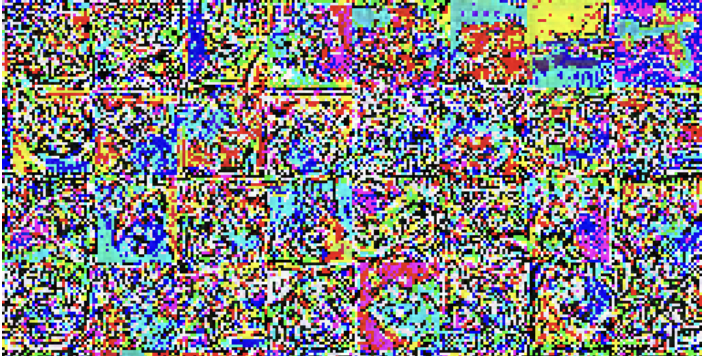}
    \caption{\footnotesize{The reconstructed samples using $\mathcal{X}^{(U)}$.}}
    \label{fig:dgattack_recon}
\end{subfigure}
\vfill
\begin{subfigure}{.8\linewidth}
    \centering
    \captionsetup{justification=centering}
    \includegraphics[width=.8\linewidth]{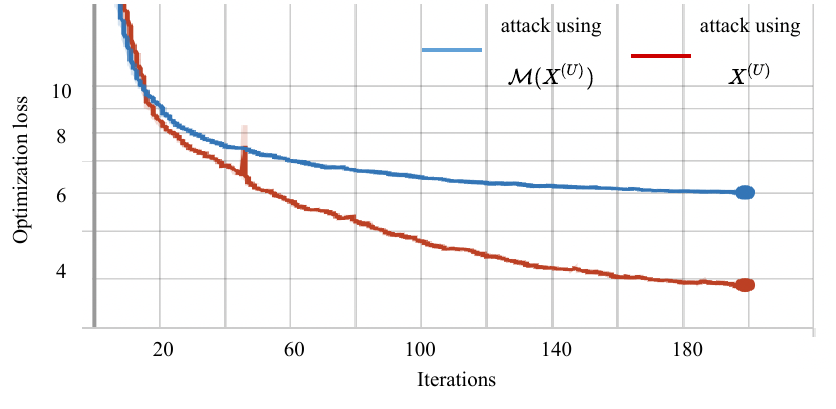}
    \caption{\footnotesize{Optimization loss in the gradient inversion attack.}}
    \label{fig:dgattack_loss}
\end{subfigure}
\caption{The reconstructed data samples using the gradients inversion attack with residual data $\mathcal{M}(\XUT)$ as prior information, and ResNet-18 as the target model are shown.}
\label{fig:dgattackperf}
\end{figure}

\section{Discussion}\label{sec:discussion}
\AsymML{} in this paper reduces the computation and memory costs in TEEs without significant privacy leaks. However, due to the additional communication overhead between GPUs and TEEs, the actual performance improvement does not exactly match the expectation in Fig.~\ref{fig:complexity}. As Fig.~\ref{fig:runbreakdown_vgg} in Appendix \ref{sec:runtimemore} shows, this additional overhead even dominates the runtime in layers with large-size  features. Architectures such as unified memory access (UMA) \cite{hUMA,AppleUMA} can potentially resolve this issue by relieving CPU-GPU communication burden, where \AsymML{} can also be leveraged.  

\AsymML{} is designed to protect privacy under the assumption that TEEs are secure against potential attacks such as side channel attacks \cite{SideChannelAttack}.  The case that side channel attacks breach TEE's security \cite{SGXspectra} based on information gained from the implementation of a computer system (e.g. power consumption, electromagnetic leaks) are out of the scope of this work. Finally, it is worth noting that \AsymML{} is compatible with security updates in hardware such as Intel SGX \cite{IntelSGX} and RISC-V Sanctum \cite{Sanctum}.

\section{Conclusion}\label{sec:conclusion}
In this paper, we have proposed an asymmetric decomposition framework, \AsymML{}, to decompose DNN models and offload computations onto trusted and untrusted fast hardware. 
The trusted part aims to preserve the main information in the data with manageable computation cost, and the untrusted part undertakes most computations. In such a way, \AsymML{} makes the best use of each platform in a heterogeneous setting.
We have then presented theoretical analysis showing that the low-rank structure is preserved in NNs, and \AsymML{} achieves $(\epsilon,\delta)$-DP guarantee by adding a Gaussian noise to $\XUT$. Our extensive experiments show that AsymML achieves gain up to $7.6\times$ in training. We also show that \AsymML{} provides strong privacy protection under model and gradient inversion attacks.

\subsection*{Acknowledgements}
This material is based upon work supported by Defense Advanced Research Projects Agency (DARPA) under Contract No. FASTNICS HR001120C0088, NSF grants CCF-1703575, CCF-1763673, CNS-2002874, and a gift from Intel/Avast/Borsetta via the PrivateAI institute.
We would like to also acknowledge the valuable feedback and discussions by Dr. Matthias Schunter from Intel.


\appendix
\section{Proofs of All Lemmas and Theorems}
\label{app:A}
In this appendix, we provide the proofs of all lemmas and theorems as follows. 
\begin{enumerate}
    \item Section \ref{subsec:entropy-bounds} presents the proof of Lemma \ref{lemma:CEntropy}.
    \item Section \ref{subsec:no_principle_channels} provides the proof of Theorem~\ref{theorem::CEntropy}. 
    \item Section \ref{subsec:output_entropy} provides the proof of Theorem \ref{theorem:CEntropy_conv1x1}. 
    \item Section \ref{subsec:output_entropy_k} provides the proof of Theorem \ref{theorem:CEntropy_convkxk}. 
    \item Section \ref{subsec:entropy_batchnorm} presents the proof of Theorem \ref{theorem:CEntropy_batchnorm}. 
    \item Section \ref{subsec:entropy_pooling} presents the proof of Theorem \ref{theorem:CEntropy_pooling}. 
    \item Section \ref{subsec:dp} provides the proof of Theorem \ref{theorem:dp}.
\end{enumerate}
\subsection{Proof of Lemma~\ref{lemma:CEntropy}}
\label{subsec:entropy-bounds}
\begin{proof}
Let $Q = \sum_{i=1}^{\ci} s_i$ and $P = \sum_{i=1}^{\ci} s_i^2$, in which $s_1 \geq s_2 \geq \cdots \geq s_N$. We note that $Q$ satisfies
\begin{align}
\label{eqn:P-Q-Ineq}
\sqrt{P}\leq Q \leq \sqrt{\ci P},
\end{align}
where the left equality holds when all singular values are zeros except $s_1$ and the right equality holds when all singular values are equal. 
According to Definition \ref{def:CEntropy}, we have 
\begin{align}
\label{eqn:entropy-Q-P}
\mu_X &= -\log \left(\sum_{i=1}^{\ci} \bar{s}_i^2(X) \right) = -\log \left(\sum_{i=1}^{\ci} \frac{s_i^2}{Q^2}\right) \notag \\&= -\log(\sum_{i=1}^{\ci} s_i^2) + \log Q^2 \notag =
 2\log Q - \log P.
\end{align}

\noindent Finally, from above equation, we have $0 \leq \mu_X \leq \log \ci$.
\end{proof}

\subsection{Proof of Theorem~\ref{theorem::CEntropy}}
\label{subsec:no_principle_channels}
\begin{proof}
Let $Q = \sum_{i=1}^{\ci} s_i$ and $P = \sum_{i=1}^{\ci} s_i^2$, and assume that $s_1 > s_2 > \cdots > s_N$. According to the assumption, the $i$-th singular value is given as $s_i=a\cdot b^{i-1}$, where $a>0, 0<b<1$. Hence, we have 
\begin{align*}
  &Q = \sum_{i=1}^{\ci}s_i = a\cdot \sum_{i=1}^{\ci}b^{i-1}=\frac{1-b^{\ci}}{1-b}, \notag \\
  &P =\sum_{i=1}^{\ci} s_i^2 = a^2\cdot \sum_{i=1}^{\ci}b^{2(i-1)}=\frac{1-b^{2\ci}}{1-b^2}, \notag
\end{align*}
\begin{align*}
  &2^{\mu_X} =\frac{Q^2}{P} =  \frac{(1+b)(1-b^{\ci})}{(1-b)(1+b^{\ci})}.
\end{align*}

\noindent Let $\eta = \frac{\sum_{i=1}^{\cii}s_i^2}{\sum_{j=1}^{\ci}s_j^2}$, which can be lower-bounded as follows
\begin{align*}
\eta (b, N) &= \frac{1-b^{2\cii}}{1-b^{2\ci}} \notag \geq \frac{1-b^{2\cdot 2^{\mu_X}}}{1-b^{2\ci}} \notag \\ &= \frac{1-b^{\frac{2Q^2}{P}}}{1-b^{2\ci}} =\frac{1-b^{\frac{2(1+b)(1-b^{\ci})}{(1-b)(1+b^{\ci})}}}{1-b^{2\ci}}.
\end{align*}
\noindent Finally by minimizing the  function $\eta(b, \ci)$, we get a minimum of $0.97$.
Therefore, with $\cii = {\left \lceil 2^{\mu_X} \right \rceil}$ principal channels, the total energy in the reconstructed data is  at least $97\%$ of the total energy in the original data $X$. 
\end{proof}
\noindent In the later theorems and proofs, we regard $\cii$ as the sufficient number of principal channel to reconstruct $X$.

\subsection{Proof of Theorem~\ref{theorem:CEntropy_conv1x1}}
\label{subsec:output_entropy}
\begin{proof}
Let $\cii = \left \lceil 2^{\mu_X} \right \rceil$, then $X_j = \sum_{\jj=1}^{\cii} a_{j,\jj}\cdot X_{\jj}^{'}$, where $X^{'}_{\jj}$ is the $\jj$-th principle channel of $X$, and $\overline{X^{'}}_{\jj}$ denote a flatten vector from $X^{'}_{\jj}$. Then,  the $i$-th output channel is given by 
\begin{align*}
Y_i = \sum_{j=1}^{\ci} W_{i,j}\circledast X_j =  \sum_{\jj=1}^{\cii} (\sum_{j=1}^{\ci} W_{i,j}\cdot a_{j,\jj})\circledast X_{\jj}^{'}.
\end{align*}
All channels in $Y$ can be then written as follows
\begin{align*}
    Y &= \left \{ Y_1, \cdots, Y_{\co} \right \} =  \footnotesize{\sum_{\jj=1}^{\cii} \left \{ (\sum_{j=1}^{\ci} W_{1,j}\cdot a_{j,\jj})\circledast X_{\jj}^{'}, \cdots\right \}}.
\end{align*}

\noindent Let 
\begin{align*}
\small
&Y_{\jj}^{'} = \left \{ (\sum_{j=1}^{\ci} W_{1,j}\cdot a_{j,\jj})\circledast X_{\jj}^{'}, \cdots,  (\sum_{j=1}^{\ci} W_{\co,j}\cdot a_{j,\jj})\circledast X_{\jj}^{'} \right \},
\end{align*}
then $Y =  \sum_{\jj=1}^{\cii} Y_{\jj}^{'}$. Since $W_{i,j}$ is a $1\times 1$ kernel, $Y_{\jj}^{'}$ can be written as

\begin{align*}
&Y_{\jj}^{'} = \left \{ (\sum_{j=1}^{\ci} W_{1,j}\cdot a_{j,\jj})\cdot X_{\jj}^{'}, \cdots,  (\sum_{j=1}^{\ci} W_{\co,j}\cdot a_{j,\jj})\cdot X_{\jj}^{'} \right \}.
\end{align*}

\noindent For any two principal channels, $X_{\jj_1}^{'}, X_{\jj_2}^{'}, \jj_1 \neq \jj_2$,  $\left <\overline{X^{'}}_{\jj_1}, \overline{X^{'}}_{\jj_2} \right > = 0$. Therefore, $Y_{\jj_1}^{'}, Y_{\jj_2}^{'}$ are constructed by two orthogonal channels.
Y only has at most $\cii$ principle channels. Therefore, $\mu_Y \leq \log \cii = \log \left \lceil 2^{\mu_X} \right \rceil.$

\end{proof}


\subsection{Proof of Theorem~\ref{theorem:CEntropy_convkxk}}
\label{subsec:output_entropy_k}
To prove Theorem~\ref{theorem:CEntropy_convkxk}, we first need a lemma to bound SVD-channel entropy for patches in all channels $\hat{X}_{:,q,r}$ is almost the same as that of the original data $X$
Then according to Lemma~\ref{lemma:CEntropy_along_channels}, we can bound SVD-channel entropy after $k\times k$ convolutions.
\begin{lemma}
\label{lemma:CEntropy_along_channels}
Given an input $X\in R^{\ci\times\Hi\times\Wi}$ with SVD-channel entropy $\mu_X$, $k\times k$ kernels, then for $\forall 1\leq q,r\leq k$, SVD-channel entropy in $\hat{X}_{:,q,r}$ satisfies:
\begin{center}
    $\mu_{\hat{X}_{:,q,r}} \leq \log \ceiling{2^{\mu_X}}$
\end{center}
\end{lemma}
\begin{proof}
Flatten $\hat{X}_{:,q,r}$, $X$ as $\overline{\hat{X}}_{:,q,r}$ and $\overline{X}$. 
Then, $\overline{\hat{X}}_{:,q,r}$ can be regarded as $\overline{X}$ with at most  $k^2-(\frac{k+1}{2})^2$ columns reset as zeros. We use a set $S_0$ to denote these columns. 
Let $\cii = \log \left \lceil 2^{\mu_X} \right \rceil$, then each row in $\overline{X}$ can be written as $ \overline{X}_j = \sum_{\jj=1}^{\cii} a_{j,\jj} \overline{X}_{\jj}^{'}$. Therefore, each row in $\overline{\hat{X}}_{:,q,r}$ can be written as $\overline{\hat{X}}_{j,q,r} = \sum_{\jj=1}^{\cii} a_{j,\jj} \overline{\hat{X}}_{\jj,q,r}^{'}$, where $\overline{\hat{X}}_{\jj,q,r}^{'}$ is the same as $\overline{X}_{\jj}^{'}$ except for values in columns $S_0$ are zeros. Therefore, $\overline{\hat{X}}_{:,q,r}$ has at most $\cii$ principle components. Namely, $\hat{X}_{:,q,r}$ has at most $\cii$ principle channels. Hence,  $\mu_{\hat{X}_{:,q,r}} \leq \cii = \log \left \lceil 2^{\mu_X} \right \rceil$.
\end{proof}

With Lemma \ref{lemma:CEntropy_along_channels}, we prove Theorem \ref{theorem:CEntropy_convkxk} as follows:
\begin{proof}
According to Lemma~\ref{lemma:CEntropy_along_channels}, 
for the $(q,r)$-th patches in all channels, $\hat{X}_{:,q,r}, \forall 1\leq q,r\leq k$ can be constructed by at most $\cii$ principle channels $\left \{ \mathcal{U}_{1,q,r},\cdots,\mathcal{U}_{\cii,q,r} \right \}$ as follows

\begin{equation*}
\left\{\begin{matrix}
\hat{X}_{1,q,r} & = & a_{1,1,q,r}\mathcal{U}_{1,q,r}+ \cdots + a_{1,\cii,q,r}\mathcal{U}_{\cii,q,r}\\ 
\vdots & & \\
\hat{X}_{\ci,q,r}& = & a_{\ci,1,q,r}\mathcal{U}_{1,q,r}+ \cdots + a_{\ci,\cii,q,r}\mathcal{U}_{\cii,q,r}.
\end{matrix}\right.
\end{equation*}

\noindent For all patches in channel $j$, $\hat{X}_{j,:,:} \forall 1\leq j \leq \ci$ can also be constructed by principle channels $\left \{ \mathcal{V}_{j,1},\cdots,\mathcal{V}_{j,\hat{N}_j} \right \}$, where $\hat{N}_j = \left \lceil 2^{\hat{\mu}_j} \right \rceil$,
 
\begin{equation*}
\left\{\begin{matrix}
\hat{X}_{j,1,1} & = & b_{j,1,1,1}\mathcal{V}_{j,1}+ \cdots + b_{j,p_j,1,1}\mathcal{V}_{j,\hat{N}_j}\\ 
\vdots & & \\
\hat{X}_{j,k,k}& = & b_{j,1,k,k}\mathcal{V}_{j,1}+ \cdots + b_{j,p_j,k,k}\mathcal{V}_{j,\hat{N}_j}.
\end{matrix}\right.
\end{equation*} 

\noindent By combining equation (1) and (2), we can express $\hat{X}_{:,q,r}$ as follows

\begin{equation*}
\label{eqn:thm3}
\left\{\begin{matrix}
 \hat{X}_{1,q,r} &= a_{1,1,q,r}\mathcal{U}_{1,q,r}+ \cdots + a_{1,\cii,q,r}\mathcal{U}_{\cii,q,r} \\
  &= b_{1,1,q,r}\mathcal{V}_{1,1}+ \cdots + b_{1,p_j,q,r}\mathcal{V}_{1,\hat{N}_1}\\ 
 \vdots &\\ 
 \hat{X}_{\ci,q,r} &= a_{\ci,1,q,r}\mathcal{U}_{1,q,r}+ \cdots + a_{\ci,\cii,q,r}\mathcal{U}_{\cii,q,r} \\
 & = b_{\ci,1,q,r}\mathcal{V}_{\ci,1}+ \cdots + b_{\ci,p_j,q,r}\mathcal{V}_{\ci,\hat{N}_{\ci}}.
\end{matrix}\right.
\end{equation*}
In (\ref{eqn:thm3}), the coefficients are obtained from SVD, therefore the determinant of any sub coefficient matrix is not zero. Hence, at most $\cii$ sub-equations are needed to derive $\mathcal{U}_{:,q,r}$ from $\mathcal{V}_{:,:}$.

\noindent If we pick the equations with least number of principal channels $\mathcal{V}$, then each channel in $\hat{X}$ can be fully constructed by the selected $\mathcal{V}$. At most, the total number of principal channels $\mathcal{V}$ needed is $\sum_{j=1}^{\cii}\hat{N}_j$. Then the total number of principle channels needed to construct $\hat{X}$ is at most $\sum_{j=1}^{\cii}\hat{N}_j$.

\noindent 
Therefore, $\mu_{\hat{X}} \leq \log (\sum_{j=1}^{\cii}\hat{N}_j) = \log (\sum_{j=1}^{\cii}\left \lceil 2^{\hat{\mu}_j} \right \rceil)$. When $\hat{\mu}_j$ are close, $\mu_{\hat{X}} \doteq \eta + \bar{\mu}$, where $\bar{\mu}$ is the average of for $\hat{\mu}_j$. Finally, according to Theorem~\ref{theorem:CEntropy_conv1x1}, $\mu_Y \leq \log \left \lceil 2^{\mu_{\hat{X}}} \right \rceil) \leq \log(\sum_{j=1}^{\cii}\left \lceil 2^{\hat{\mu}_j} \right \rceil)$.
\end{proof}

\subsection{Proof of Theorem~\ref{theorem:CEntropy_batchnorm}}
\label{subsec:entropy_batchnorm}
\begin{proof}
Let $\cii= \left \lceil 2^{\mu_X} \right \rceil$, then input $X_j = \sum_{\jj=1}^{\cii} a_{j,\jj}\cdot X_{\jj}^{'}$, where $X^{'}$ is the principle channels of $X$. With batch normalization operator, we have
\begin{align*}
Y_j = \frac{X_j - E[X_j]}{\sqrt{V[X_j]+\epsilon}} \cdot \gamma_i + \beta_i,
\end{align*}
where $E[X_j]$ and $V[X_j]$ is the mean and variance in channel $j$ across batches, $\gamma_j$ and $\beta_j$ are learnable parameters in BatchNorm layers and $\epsilon$ is a constant for numerical stability. We can then re-write $Y_i$ in terms of $X_{\jj}^{'}$ as 
\begin{align*}
    Y_i &= \frac{\sum_{\jj=1}^{\cii} a_{j,\jj}\cdot X_{\jj}^{'} - E[X_j]}{\sqrt{V[X_j]+\epsilon}}\cdot \gamma_j + \beta_j \\
    &= \sum_{\jj=1}^{\cii} \frac{\gamma_j a_{j,\jj}}{\sqrt{V[X_j]+\epsilon}} X_{\jj}^{'} - \frac{\gamma_j E[X_j]}{\sqrt{V[X_j]+\epsilon}} + \beta_j\\
\end{align*}
\begin{align*}
    =\sum_{\jj=1}^{\cii} \frac{\gamma_j a_{j,\jj}}{\sqrt{V[X_j]+\epsilon}} X_{\jj}^{'} - (\frac{\gamma_j E[X_j]}{\sqrt{V[X_j]-\epsilon}} - \beta_j)\cdot \bm{1}
\end{align*}
\noindent Therefore, the output channel $Y_j$ can be reconstructed by at most $\cii$ principal channels in $X$, plus a constant vector $\bm{1}$. Hence $\mu_Y\leq \log\ceiling{2^{\mu_X}} + 1$.
\end{proof}

\subsection{Proof of Theorem~\ref{theorem:CEntropy_pooling}}
\label{subsec:entropy_pooling}
\begin{proof}
Let $\cii = \left \lceil 2^{\mu_X} \right \rceil$, then input $X_j = \sum_{\jj=1}^{\cii} a_{j,\jj}\cdot X_{\jj}^{'}$, where $X^{'}$ is the principle channels of $X$. With $k\times k$ average operator, we have
\begin{align*}
\small
&Y_{i}(h,w) 
= \frac{1}{k^2}\sum_{h^{'}=1}^{k}\sum_{w^{'}=1}^{k}X_{i}((h-1)k+h^{'},(w-1)k+w^{'}) \notag \\
&= \frac{1}{k^2}\sum_{h^{'}=1}^{k}\sum_{w^{'}=1}^{k}\sum_{\jj=1}^{\cii}a_{i,\jj}X_{\jj}^{'}((h-1)k+h^{'},(w-1)k+w^{'}) \notag \\
&=\sum_{\jj=1}^{\cii}a_{i,\jj}\cdot \frac{\sum_{h^{'}=1}^{k}\sum_{w^{'}=1}^{k}X_{\jj}^{'}((h-1)k+h^{'},(w-1)k+w^{'})}{k^2}
\end{align*}

\noindent Therefore, each channel in $Y$ can be seen as an accumulation of $\cii$ principle channels obtained by conducting average pooling on $X^{'}$. Therefore, $Y$ can be constructed by at most $\cii$ principle channels. Hence, $\mu_Y \leq  \log \cii = \log \left \lceil 2^{\mu_X} \right \rceil$
\end{proof}

\subsection{Proof of Theorem \ref{theorem:dp}}\label{subsec:dp}
\begin{proof}
\newcommand{\MXu}{\mathcal{M}(X^{(U)})}
\newcommand{\MXun}{\mathcal{M}(X^{'(U)})}
\newcommand{\abs}[1]{\left| #1 \right|}
Due to the random sampling process,  two cases might arise: 1) if $X^i \notin \mathcal{S}$, $\norm{\mathcal{S}^{(U)}-\mathcal{S}^{'(U)}} = 0$; 2) if $X^i \in \mathcal{S}$, $\norm{\mathcal{S}^{(U)}-\mathcal{S}^{'(U)}} \leq \Delta_2$.

\noindent Since $\mathcal{S}$ and $\mathcal{S}^{'}$ can only differ at one position, without loss of generality, we assume $X$ and $X^{'}$ are the data that might be different in $\mathcal{S}$ and $\mathcal{S}^{'}$. According to data decomposition, $X = \XT + \XUT$, and $X^{'} = X^{'(T)}+X^{'(U)}$. With noise $z = \mathcal{N}(0,\sigma^2I)$ added to $\XUT$ and $X^{'(U)}$, we have 
\begin{equation*}
    \mathcal{M}(\XUT) = \XUT + z,\quad \mathcal{M}(X^{'(U)}) = X^{'(U)} + z.
\end{equation*}

\noindent Given $0<\epsilon\leq 1, \delta>0$, we next focus on deriving the tail bound $\delta$. We define $C$ as 
\begin{equation*}
    C = \log{\frac{\Pr(\MXu=\XUT+z)}{\Pr(\MXun=\XUT+z)}}.
\end{equation*}
Then, we have
\begin{align*}
    \Pr(\abs{C} \geq \epsilon) = 
    & \Pr(\XUT\neq\XUTn)\cdot \\ &\Pr(\abs{C}\geq \epsilon|\XUT\neq\XUTn).
\end{align*}

\noindent Given $\XUT\neq\XUTn$, let $v = \XUT-\XUTn$,  $C$ can be written as
\begin{align*}
    C &= \log{\frac{\exp{(-\norm{z}^2/2\sigma^2)}}{\exp{(-\norm{z+v}^2/2\sigma^2)}}} \\
    &= -\frac{1}{2\sigma^2} ( \norm{z}^2 - \norm{z+v}^2 ) \\
    &=\frac{1}{2\sigma^2}(2\left\langle z, v \right\rangle + \norm{v}^2),
\end{align*}
where $\left\langle ., . \right\rangle$ denotes the sum of element-wise product. It is easy to verify that $C$ given  $\XUT\neq\XUTn$ is a Gaussian random variable with mean $\frac{\norm{v}^2}{2\sigma^2}$ and variance $\frac{\norm{v}^2}{\sigma^2}$. Since $q = \Pr(\XUT\neq\XUTn)$, $\Pr(\abs{C}\geq \epsilon) \leq \delta$ is equivalent to $\Pr(\abs{C}\geq \epsilon|\XUT\neq\XUTn) \leq \frac{1}{q}\delta=\delta^{'}$. \\
Similar as the classical Gaussian mechanism, we set $\sigma=\frac{t\Delta_2}{\epsilon}$. Since $\norm{v}\leq\Delta_2$, $\Pr(\abs{C}\geq \epsilon|\XUT\neq\XUTn)$ is equivalent to 
\begin{align*}
    &\Pr\left(\abs{STD(C|\XUT\neq\XUTn)}\geq\frac{\epsilon\sigma}{\norm{v}}-\frac{\norm{v}}{2\sigma^2}\right) \\
    &\Leftrightarrow \\
    &\Pr\left(\abs{STD(C|\XUT\neq\XUTn)}\geq t-\frac{\epsilon}{2t}\right)
\end{align*}
$STD$ denotes the standardization of a distribution.
Following the classical Gaussian mechanism \cite{dfbasic}, if we set $t=\sqrt{2\log{(1.25/\delta^{'})}}=\sqrt{2\log{(1.25q/\delta)}}$, we  get
\begin{align*}
    \Pr(\abs{C}\geq\epsilon) &= 
    \Pr(\abs{STD(C|\XUT\neq\XUTn)}\geq t-\frac{\epsilon}{2t}) \\ &\leq \delta.
\end{align*}

\end{proof}

\section{Model Architecture for MI Attack}
\label{sec:miArchitecture}

In this appendix, we provide the model architectures of $G$ and $D$ in Table \ref{tab:generator} and Table \ref{tab:discriminator},  respectively.

\begin{table}[!htb]
\centering
\begin{tabular}{llllll}
\toprule \\
Type & Kernel & Channels & Stride & Padding & Output \\
\midrule
\multicolumn{6}{c}{B1} \\
\midrule
Conv & 3 & 32 & 1 & 1 & $32\times 32$ \\
Conv & 3 & 64 & 2 & 1 & $16\times 16$ \\
Conv & 3 & 128 & 1 & 1 & $16\times 16$ \\
Conv & 3 & 128 & 2 & 1 & $8\times 8$ \\
Conv & 3 & 128 & 1 & 1 & $8\times 8$ \\
\midrule
\multicolumn{6}{c}{B2} \\
\midrule
DeConv & 4 & 256 & 1 & 0 & $4\times 4$ \\
DeConv & 4 & 128 & 2 & 1 & $8\times 8$ \\
\midrule
\multicolumn{6}{c}{Decoder} \\
\midrule
DeConv & 4 & 128 & 2 & 1 & $16\times 16$ \\
DeConv & 4 & 64 & 2 & 1 & $32\times 32$ \\
Conv & 3 & 32 & 1 & 1 & $32\times 32$ \\
Conv & 3 & 3 & 1 & 1 & $32\times 32$ \\
\bottomrule
\end{tabular}
\caption{The generator architecture (B1, B2, and the decoder) is shown. B1 is used to extract the features from the residual data $\XUT$, while B2 is used to generate the latent features. The  decoder then combines these features, and reconstructs the output.}
\label{tab:generator}
\end{table}
\vspace{-10 pt}
\begin{table}[!htb]
\centering
\begin{tabular}{llllll}
\toprule \\
Type & Kernel & Channels & Stride & Padding & Output \\
\midrule
Conv & 3 & 32 & 2 & 1 & $16\times 16$ \\
Conv & 3 & 64 & 2 & 1 & $8\times 8$ \\
Conv & 3 & 128 & 2 & 1 & $4\times 4$ \\
Conv & 3 & 256 & 2 & 1 & $2\times 2$ \\
Conv & 3 & 256 & 2 & 1 & $1\times 1$ \\
\bottomrule
\end{tabular}
\caption{The discriminator architecture is shown. Batch normalization and ReLU is applied after every convolution layer, similar to the generator.}
\label{tab:discriminator}
\end{table}

\section{SVD-Channel Entropy in NN Models} \label{sec:SEntropyModels}
In this appendix, we empirically illustrate how the SVD-channel entropy changes across the layers in DNNs.  

\noindent In Fig.~\ref{fig:entropy}, we show the average SVD-channel entropy and the required number of principal channels to approximate intermediate data across all batches in VGG-19/ImageNet with a randomly initialized model.
We note the following two key observations.
First, as shown in Figure \ref{fig:entropy:entropy}, the SVD-channel entropy before and after the ReLU layers barely changes. 
As for Pooling (MaxPooling) layers, the SVD-channel entropy usually decreases since the features are down-sampled in these layers.
Second, from Figure \ref{fig:entropy:channels}, we observe that the required number of principal channels to reconstruct the original intermediate features is much less than the number of original kernels. Thus, these observations provide strong evidence for the existence of low-rank structure in NNs. 

\begin{figure}[htb!]
\centering
\begin{subfigure}{.85\linewidth}
    \centering
    \captionsetup{justification=centering}
    \includegraphics[width=.9\linewidth]{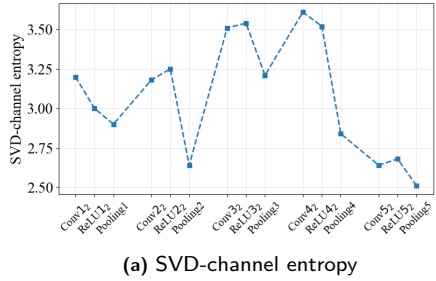}
    \caption{\footnotesize{SVD-channel entropy}}
    \label{fig:entropy:entropy}
\end{subfigure}
\vfill
\begin{subfigure}{.9\linewidth}
    \centering
    \captionsetup{justification=centering}
    \includegraphics[width=.9\linewidth]{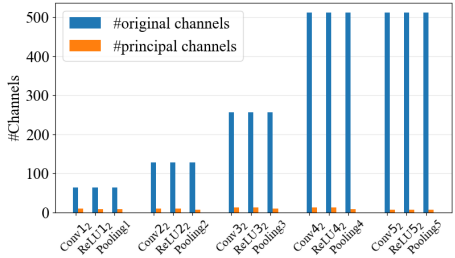}
    \caption{\footnotesize{Number of principal channels vs. original channels}}
    \label{fig:entropy:channels}
\end{subfigure}
\caption{SVD-channel entropy in VGG-19/ImageNet.}
\label{fig:entropy}
\vspace{-.4cm}
\end{figure}

\section{Approximate SVD Evaluation}\label{appx:approxsvd}
In this appendix, we further evaluate the approximated light SVD in Algorithm \ref{alg::lightSVD}. \\
To measure the efficacy of Algorithm \ref{alg::lightSVD}, we calculate the total energy in the remaining data $\Xin{i}$ after extracting the $i$ most principal channels as shown in Equation (\ref{eq:X_svd_light}). It is worth noting that minimizing such residual energy is also one of the objectives of the original SVD algorithm. 
To give a better comparison, we use relative energy, $\left \| \Xin{i }\right \|^2_F / \left \| \overline{X}\right \|^2_F$ and compare with SVD.
Figure \ref{fig:resenergy} shows the relative residual energy by performing SVD and our approximated algorithm in outputs after the first and the second convolution layer in VGG-16. The results are obtained by averaging multiple mini-batches using a randomly initialized model.
We observe that the extracted principal channels using the approximated SVD contain almost the same energy as using SVD. 
Therefore, such an approximation algorithm captures low-rank components in data as using the original SVD algorithm. 
\begin{figure}[htb!]
\centering
\begin{subfigure}{.8\linewidth}
    \centering
    \captionsetup{justification=centering}
    \includegraphics[width=.8\linewidth]{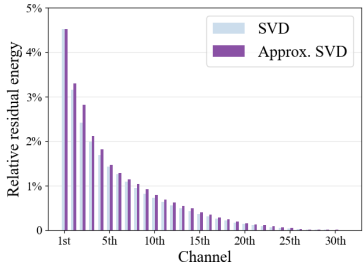}
    \caption{\footnotesize{1st convolution layer.}}
    \label{fig:res:conv1}
\end{subfigure}
\vfill
\begin{subfigure}{.8\linewidth}
    \centering
    \captionsetup{justification=centering}
    \includegraphics[width=.8\linewidth]{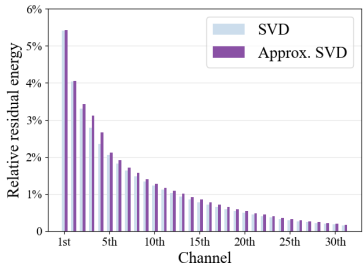}
    \caption{\footnotesize{2nd convolution layer.}}
    \label{fig:res:conv2}
\end{subfigure}
\caption{Energy in residuals after 1st and 2nd convolution layer is shown in VGG-16 using SVD and the approximated SVD in Algorithm \ref{alg::lightSVD}. The extracted principal channels using approximated SVD contain almost the same energy as using SVD.}
\label{fig:resenergy}
\vspace{-.4cm}
\end{figure}

\section{Runtime Breakdown of VGG-19 and ResNet-34}
\label{sec:runtimemore}
In this appendix, we provide the runtime breakdown of VGG-19 and ResNet-34 as shown in Figure \ref{fig:runbreakdown_vgg}. 
Similar to in Figure \ref{fig:runbreakdown}, the communication between SGX and GPUs dominates the runtime in the early convolution layers, especially in backward passes, while it becomes marginal in the later layers.
\begin{figure}[htb!]
\centering
\begin{subfigure}{.9\linewidth}
    \centering
    \captionsetup{justification=centering}
    \includegraphics[width=.85\linewidth]{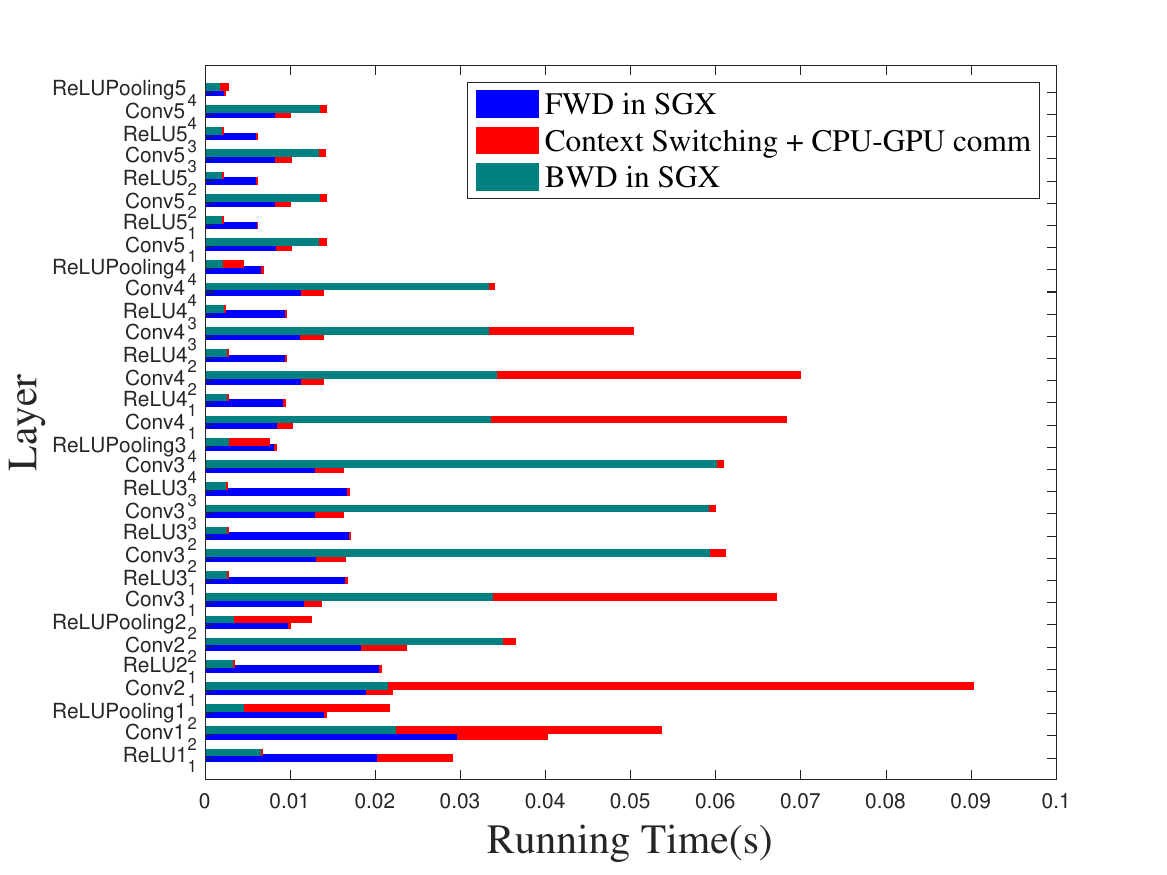}
    \caption{\footnotesize{The runtime of VGG-19 is shown.}}
    \label{fig:runbreakdown_vgg16}
\end{subfigure}
\vfill
\begin{subfigure}{.9\linewidth}
    \centering
    \captionsetup{justification=centering}
    \includegraphics[width=.8\linewidth]{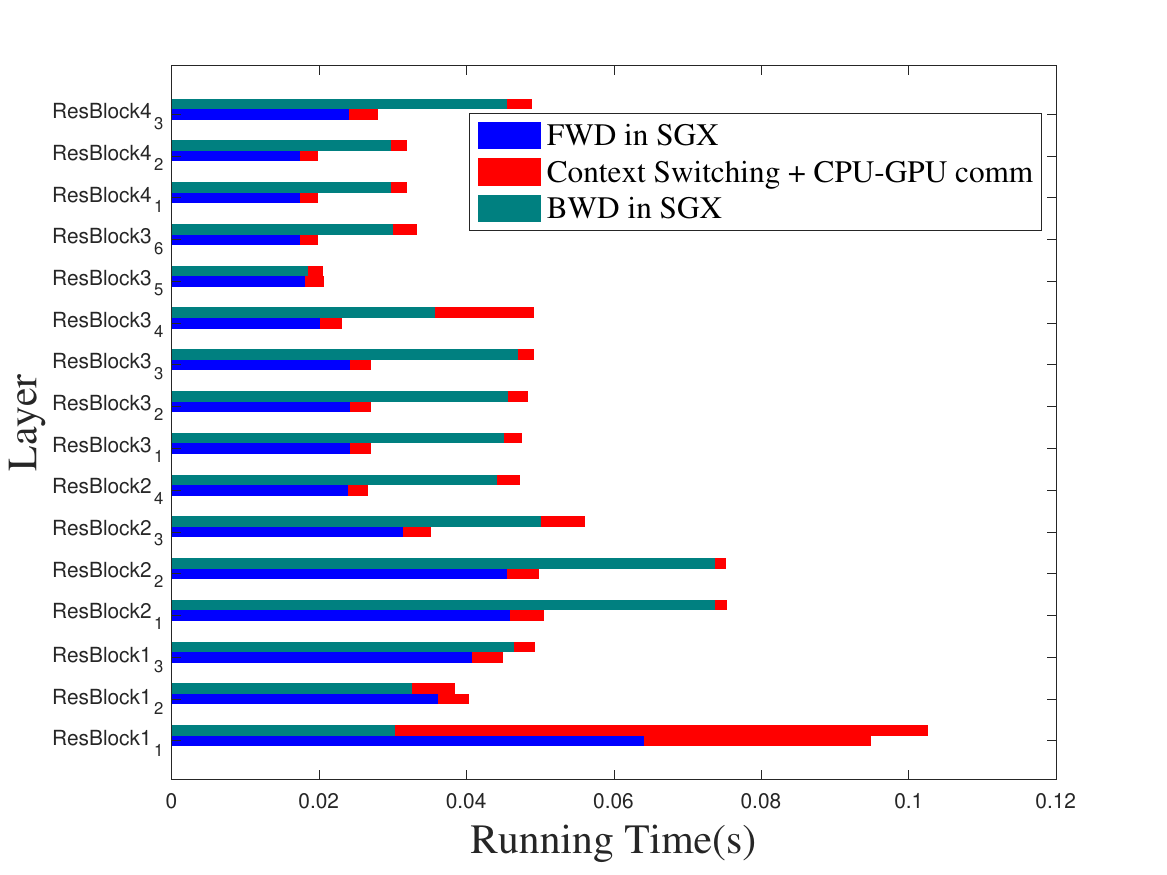}
    \caption{\footnotesize{The runtime of ResNet-34 is shown.}}
    \label{fig:runbreakdown_vgg19}
\end{subfigure}
\caption{The runtime breakdown in VGG-19 and ResNet-34 is shown.Time in data movement is relatively high in early convolutional layers, which then becomes marginal in later  layers.}
\label{fig:runbreakdown_vgg}
\end{figure}

\end{document}